\RequirePackage{amsmath}
\documentclass[12pt,letterpaper,DIV=calc,DIV=14,abstract=off,
]{scrartcl}

\usepackage[utf8]{inputenc}
\usepackage{overlock}
\usepackage[varqu,varl]{inconsolata}
\usepackage[exscale]{ccfonts}

\usepackage[cal=euler]{mathalfa}
\usepackage[T1]{fontenc}
\usepackage{amsmath,amsfonts,amsthm,verbatim,amssymb}
\allowdisplaybreaks[4]

\usepackage{mathtools}
\usepackage{dsfont}
\usepackage{xcolor}
\usepackage{array}
\usepackage{graphicx}
\usepackage[numbers]{natbib} 
\renewcommand{\cite}{\citep}
\usepackage{tikz}
\usetikzlibrary{arrows,backgrounds,shapes}
\usetikzlibrary{hobby}
\usetikzlibrary{arrows.meta}
\usepackage{nicefrac}

\usepackage[inline]{enumitem}
\usepackage[hidelinks]{hyperref}

\newcolumntype{L}[1]{>{\raggedright\let\newline\\\arraybackslash\hspace{0pt}}m{#1}}
\newcolumntype{C}[1]{>{\centering\let\newline\\\arraybackslash\hspace{0pt}}m{#1}}
\newcolumntype{R}[1]{>{\raggedleft\let\newline\\\arraybackslash\hspace{0pt}}m{#1}}

\DeclareMathOperator*{\argmax}{argmax}

\newcommand{\mN}{\mathcal{N}}

\renewcommand{\hat}{\widehat}

\renewcommand{\breve}{\tilde}

\newtheorem{theorem}{Proposition}
\newtheorem{definition}{Definition}
\newtheorem{lemma}{Lemma}

\newtheorem{fact}{Fact}

\newcommand\blfootnote[1]{%
	\begingroup
	\renewcommand\thefootnote{}\footnote{#1}%
	\addtocounter{footnote}{-1}%
	\endgroup
}

\begin{document}
	
\title{Distributed Mechanism Design with Learning Guarantees}
\author{Abhinav Sinha and Achilleas Anastasopoulos\\
	\normalsize{EECS Department, University of Michigan, Ann Arbor.} \\
	\normalsize{\texttt{\{absi,anastas\}@umich.edu}}
	\blfootnote{This work is supported in part by NSF grant ECCS-1608361.}
}

\date{\normalsize\today} 
\maketitle
	
\begin{abstract}
Mechanism design for fully strategic agents commonly assumes broadcast nature of communication between agents of the system. Moreover, for mechanism design, the stability of Nash equilibrium (NE) is demonstrated by showing convergence of specific pre-designed learning dynamics, rather than for a class of learning dynamics. In this paper we consider two common resource allocation problems: sharing $ K $ infinitely divisible resources among strategic agents for their private consumption (private goods), and determining the level for an infinitely divisible public good with $ P $ features, that is shared between strategic agents. For both cases, we present a distributed mechanism for a set of agents who communicate through a given network. In a distributed mechanism, agents' messages are not broadcast to all other agents as in the standard mechanism design framework, but are exchanged only in the local neighborhood of each agent. The presented mechanisms produce a unique NE and fully implement the social welfare maximizing allocation. In addition, the mechanisms are budget-balanced at NE. It is  also shown that the mechanisms induce a game with contractive best-response, leading to guaranteed convergence for all learning dynamics within the Adaptive Best-Response dynamics class, including dynamics such as Cournot best-response, $ k- $period best-response and Fictitious Play. We also present a numerically study of convergence under repeated play, for various communication graphs and learning dynamics. 
\end{abstract}

\section{Introduction}

Mechanism design has been studied extensively in Economics~\cite{hurwicz2006,vohra2011mechanism,borgers2015book} and Engineering literature~\cite{kelly,basar,hajekvcg,johari2009efficiency,rahuljain,demos,bhattacharya2016,SiAn_multicast_tcns}. Most of the mechanisms presented in literature for the case of strategic agents have two (unrelated) drawbacks that are an impediment to  their applicability in real-world scenarios. The first drawback is the assumption that the underlying communication structure between agents is a broadcast one. Almost all mechanisms in the strategic setting define allocation (and taxes/subsidies) in such a way that it requires each agent to broadcast their message, i.e., the allocation function is a function of all users' messages.  The second drawback is the dynamic stability of Nash equilibrium (NE) of the game induced by the mechanism. Results on this, for most of the presented mechanisms in the literature, are either non-existent or restricted to narrow definition(s) of dynamic stability. As a result it is not clear if, and under what conditions on the dynamics, the NE is reached.
	
Regarding the first issue mentioned above, until now little consideration has been given to designing incentives for strategic agents for whom communication is restricted by a network. The motivation for this modeling consideration comes from the literature on distributed optimization,~\cite{nedic2009,boyd2011,duchi2012,scutari2016}, where algorithms are designed for global consensus between distributed non-strategic agents, who possess local information and who communicate locally on a network. In mechanism design, a designer designs incentives such that strategic agents ``agree'' to reveal their relevant private (local) information truthfully. Thus, with the above motivation in mind, a natural question is to ask: can incentives be designed for strategic agents with local private information who communicate locally on a network? One expects that attaining consensus between strategic agents only becomes harder to achieve when message exchange is restricted by the network structure.
This aforementioned issue is not to be confused with that of distributed optimization where the problem of local exchange of information has been addressed and to a large extend solved~\cite{nedic2009,boyd2011,duchi2012,scutari2016}. Neither should it be confused with the local public goods models (e.g.,~\cite{local2003,demosshruti}), where each agent's utility in the model is already assumed to only depend on his/her neighbors' allocations.
	
Regarding the second drawback mentioned above, there is a long line of work investigating stability of NE through learning in games~\cite{milgrom1990,young2004strategic,fudenberglearning}. Theoretically, the notion of NE applies to perfect information settings, i.e., where each agent knows the utility of every other agent. However, for most of the mechanism design works in the literature where NE is used as the solution concept, the models are not necessarily restricted to the perfect information setting. Indeed, in an informationally and physically decentralized system it is natural to assume that agents only know their own utility and no one else's. In such cases, the robustness w.r.t. information available to agents, of any particular designed mechanism is evaluated by the learning guarantees that it can provide. Since agents can't calculate the NE offline, they are expected to learn it by repeatedly playing the induced game whilst adjusting their strategy dynamically using the past observations. The larger the class of learning dynamics that are guaranteed to converge, the more robust the mechanism.  The idea of NE being the convergent point of learning dynamics directly relates to the \emph{Evolutive} interpretation of NE,~\cite{osborne1994}, where NE even for a single-shot game  is interpreted as the stationary point of a dynamic adjustment process. The original thesis of John Nash,~\cite{nashthesis51}, too provides a similar dynamic adjustment interpretation of NE.

In this paper our objective is to design mechanisms that resolve simultaneously both the issues mentioned above. This means that the mechanisms are distributed, i.e., the allocation and tax functions (contracts) obey the communication constraints of a network and for the designed mechanisms theoretical guarantees of convergence, for a sufficiently large class of learning dynamics, can be provided.
	
The basic idea for achieving the latter this is to identify appropriate properties of games that can lead to convergence of a correspondingly selected class of learning dynamics and then design the mechanism such that the induced game possesses the identified properties. In Milgrom and Roberts~\cite{milgrom1990}, authors identify \emph{supermodularity} as a critical property of a game and prove that any learning dynamic within the \emph{Adaptive Dynamics} class is guaranteed to converge between the two most extreme Nash equilibria when the game is played repeatedly. Following this,~\cite{chen2002family} presents a mechanism for the Lindahl allocation problem such that the induced game is supermodular. Healy and Mathevet in~\cite{healy2012designing} identify contraction as a learning-relevant  property of the game and show that any learning dynamic within the \emph{Adaptive Best-Response (ABR)} dynamics class is guaranteed to converge to the unique NE. They also present a mechanism that induces a contractive game for the Walrasian and Lindahl allocation problems, under the usual broadcast information structure. The property of contraction is more stringent than supermodularity and consequently the ABR class is broader than the Adaptive Dynamics class.
	
In this paper, for both Walrasian and Lindahl allocation problems, we define a mechanism through an appropriately designed message space, allocation function and tax function. The allocation and tax functions, for any agent, depend only on the messages of his/her neighbors. The mechanism description contains certain free parameters such that for all values of these parameters the mechanism achieves its goal of full implementation in NE. For the purpose of providing learning guarantees we consider the ABR class of learning dynamics and show that by tuning the free parameters appropriately, the induced game can be made contractive.  Section~\ref{seclearning} describes some well-known learning dynamics that are part of the ABR class. With the aid of numerical analysis, we show that the designed contractive mechanism provides exponential rate of convergence for several instances of our model.
	
The structure of this paper is as follows: Section~\ref{secmodel} describes the two centralized allocation problems and their optimality conditions. Section~\ref{secpr} defines some mechanism design basics and then presents the mechanism for the private goods problem. Section~\ref{secmechpub} presents the mechanism for the public goods problem. Section~\ref{seclearning} introduces learning-related properties and contains the result of guaranteed convergence of any learning dynamic within the ABR class. Finally, Section~\ref{seclearning} also contains a numerical study of the convergence pattern of various learning dynamics for different underlying communication graphs.

\section{Model} \label{secmodel}
	
There are $ N $ strategic agents, denoted by the set $ \mN = \{1,\ldots,N\} $. A directed communication graph $ \mathcal{G} = (\mN,\mathcal{E}) $ is given, where the vertexes correspond to the agents and an edge from vertex $ i $ to $ j $ indicates that agent $ i $ can ``listen'' to agent $ j $. It is assumed that the given graph $ \mathcal{G} $ is strongly connected. In this paper, we are interested in two different types of allocation problems: private goods and public goods, we describe each model below. In Economics literature, these are also known as Walrasian and Lindahl allocation problems~\cite{hurwicz1979outcome,MWG}.
	
\subsection{Private goods allocation problem}
	
There are $ K $ infinitely divisible goods, denoted by set $ \mathcal{K} = \{1,\ldots,K\} $, that are to be distributed among the agents. Each agent receives a utility $ v_i(x_i) $ based on the profile $ x_i =  (x_i^1,\ldots,x_i^K) $ of quantity of each good that he/she receives. Since for each agent, its utility depends only on privately consumed allocation $ x_i $ and not on other agents' allocation, this is the private goods model. 
	
It is assumed that $ v_i: \mathbb{R}^K \rightarrow \mathbb{R} $ is a continuously double-differentiable, strictly concave function that satisfies, $ \forall~k \in \mathcal{K} $, 
\begin{subequations} \label{eqetagenpr}
\begin{gather}
	-\eta < H_{kk}^{-1} + \sum_{l \in \mathcal{K},\,l\ne k} \left\vert H^{-1}_{kl} \right\vert  < 0,
	\\
	H_{kk}^{-1} < - \frac{1}{\eta},
\end{gather}
\end{subequations}
for any given $ \eta > 1 $, where $ H^{-1} $ is the inverse of the Hessian $ H = \left[ \left. (\partial^2 v_i(\cdot)) \middle/ (\partial x_i^k \partial x_i^l) \right. \right]_{k,l} $. To understand the significance of this assumption consider the case of $ K=1 $, then this condition is the same as
\begin{equation}\label{eqetapr}
	v_i^{\prime\prime}(\cdot) \in \left( -\eta,-\frac{1}{\eta} \right). 
\end{equation}
It is already assumed that $ v_i(\cdot) $ is strictly concave, the only additional imposition made by this assumption is that the second derivative of $ v_i(\cdot) $ is strictly bounded away from $ 0 $ and $ -\infty $. More generally if the utility is separable, $ v_i(x_i) = \sum_{k \in \mathcal{K}} v_{i,k}(x_i^k) $, then the condition in~\eqref{eqetagenpr} is the same as
\begin{equation}
	v_{i,k}^{\prime\prime}(\cdot) \in \left( -\eta,-\frac{1}{\eta} \right), \quad \forall~k \in \mathcal{K}.
\end{equation}
The above mentioned properties of the utility function are assumed to be common knowledge between agents and the designer. However, the utility function $ v_i(\cdot) $ itself is known only to agent $ i $ and is not known to other agents or the designer. The designer wishes to allocate available goods such that the sum of utilities is maximized subject to availability constraints, i.e., to solve the following centralized allocation problem,
\begin{subequations} \label{eqcppr} 
\begin{gather}
	x^* = \argmax_{x \in \mathbb{R}^K} \sum_{i \in \mN} v_i(x_i) 
	\\
	\label{eqconstpr}
	\text{subject to} \quad \sum_{i \in \mN} x_i^k = c_k, \quad \forall~k \in \mathcal{K},
\end{gather}
\end{subequations}
where $ c_k \in \mathbb{R} $ is the total available amount of good $ k \in \mathcal{K} $. The allocation $ x^* $ is also called the efficient allocation and it is assumed to be finite i.e., the optimization is well-defined. The efficient allocation is unique since the utilities are strictly concave. Further, the necessary and sufficient condition for optimality are
\begin{subequations}\label{eqKKTpr}
\begin{alignat}{2}
	\frac{\partial v_i(x^*)}{\partial x_i^k} &= \lambda_k^*, \quad &&\forall~k \in \mathcal{K},~\forall~i \in \mN,
	\\
	\sum_{i \in \mN} {x_i^k}^* &= c_k, \quad &&\forall~k \in \mathcal{K}, 
\end{alignat}
\end{subequations} 
where $ \big( \lambda_k^* \big)_{k \in \mathcal{K}} \in \mathbb{R}^K $ are the (unique) optimal dual variables for each constraint in~\eqref{eqconstpr}.

\subsection{Public goods allocation problem}

There is a single infinitely divisible public good with $ P $ features, with the set of features denoted by $ \mathcal{P} = \{1,\ldots,P\} $. Each agent receives a utility $ v_i(x) $ based on the quantity of the public good $ x = \left(x^p\right)_{p \in \mathcal{P}} \in \mathbb{R}^P $. Since for each agent, its utility depends on the common allocation $ x $, this is the public goods model. It is assumed that $ v_i: \mathbb{R}^P \rightarrow \mathbb{R} $ is a continuously double-differentiable, strictly concave function that satisfies~\eqref{eqetagenpr} with the Hessian $ H = \left[ \left. (\partial^2 v_i(\cdot)) \middle/ (\partial x^p \partial x^q) \right. \right]_{p,q}  $. If $ P=1 $, then this condition is the same as in~\eqref{eqetapr}.
As in the private goods model, the properties of the utility function are assumed to be common knowledge between agents and the designer. However, the utility function $ v_i(\cdot) $ itself is known only to agent $ i $ and is not known to other agents or the designer. The designer wishes to allocate the public good such that the sum of utilities is maximized, i.e., to solve the following centralized allocation problem,
\begin{align}\label{eqcppub}
	x^* = \argmax_{x \in \mathbb{R}^P} \sum_{i \in \mN} v_i(x). 
\end{align}
The allocation $ x^* $ is also called the efficient allocation and it is assumed to be finite i.e., the optimization is well-defined. The efficient allocation is unique due to strictly concave utilities. The necessary and sufficient optimality conditions can be written as
\begin{subequations}\label{eqKKTpub}
\begin{alignat}{2}
	\frac{\partial v_i(x^*)}{\partial x^p} &= {\mu_i^p}^*, \quad &&\forall~ p \in\mathcal{P},~\forall~i \in \mN,
	\\ 
	\label{eqKKTpub_b}
	\sum_{i \in \mN} {\mu_i^p}^* &= 0, \quad &&\forall~ p \in \mathcal{P},
\end{alignat}
\end{subequations} 
where $ \big( {\mu_i^p}^* \big)_{p \in  \mathcal{P}, i \in \mN} $ are the (unique) optimal dual variables.
	
In general, for the public goods problem one can also assume a seller in the system who produces the quantity $ x $ and for whom the cost of production is a known (convex) function. In this case, the social welfare maximizing allocation contains the utility of the seller as well. For keeping exposition clear and to focus on the similarities between the public and private goods problems, the seller is not considered in this model. If needed, this can accommodated in a straightforward manner.
	
In the next two sections, we present a mechanism for each of the above problems.

\section{A mechanism for the private goods problem} \label{secpr}
	
\subsection{Definitions}
A one-shot mechanism is defined by the triplet, 
\begin{equation}
	\Big( \mathcal{M} = \mathcal{M}_1 \times \cdots \times \mathcal{M}_N,\big(\widehat{x}_1(\cdot),\ldots,\widehat{x}_N(\cdot)\big), \big( \widehat{t}_1(\cdot),\ldots,\widehat{t}_N(\cdot)\big) \Big)
\end{equation}
which consists of, for each agent $ i \in \mN $, the message space $ \mathcal{M}_i $, the allocation function $ \widehat{x}_i:\mathcal{M} \rightarrow \mathbb{R}^K $ and the tax function $ \widehat{t}_i:\mathcal{M} \rightarrow \mathbb{R} $. Given a mechanism, a game $ \mathfrak{G}_{pvt} $ is setup between the agents in $ \mN $, with action space $ \mathcal{M} $ and utilities 
\begin{equation}\label{equtil} 
	u_i(m) = v_i(\widehat{x}_i(m)) - \widehat{t}_i(m).
\end{equation}
	
For this game, $ \widetilde{m} \in \mathcal{M} $ is a Nash equilibrium if 
\begin{equation}
	u_i(\widetilde{m}) \ge u_i(m_i,\widetilde{m}_{-i}), \quad \forall~m_i \in \mathcal{M}_i,~\forall~i \in \mN.
\end{equation}
The mechanism is said to \emph{fully implement} the efficient allocation if 
\begin{equation}
	\Big( \widehat{x}_1(\widetilde{m}),\ldots,\widehat{x}_N(\widetilde{m}) \Big) = x^*, \quad \forall~\widetilde{m} \in \mathcal{M}_{NE},
\end{equation} 
where $ \mathcal{M}_{NE} \subseteq \mathcal{M} $ is the set of all Nash equilibria of game $ \mathfrak{G}_{pvt} $ and $ x^* $ is the efficient allocation from~\eqref{eqcppr}. Furthermore, the mechanism is said to be \emph{budget balanced} at Nash equilibrium if
\begin{equation}
	\sum_{i\in\mN}\widehat{t}_i(\widetilde{m}) = 0, \quad \forall~\widetilde{m} \in \mathcal{M}_{NE}.
\end{equation}
Finally, we call the mechanism \emph{distributed} if for any agent $ i \in \mN $, the allocation function $ \widehat{x}_i(\cdot) $ and tax function $ \widehat{t}_i(\cdot) $ instead of depending on the entire message $ m = (m_j)_{j \in \mN} $, depend only on $ m_i $ and $ \big(m_j\big)_{j \in \mN(i)} $, i.e., agent $ i $ and his/her immediate neighbors. Here $ \mN(i) $ are all the ``out''-neighbors of $ i $ i.e., there exists an edge from $ i $ to $ j $ in graph $ \mathcal{G} $ iff $ j \in \mN(i) $ . 
	
\subsection{Mechanism}
\label{secmechpr}
	
For the purpose of maintaining clarity in the exposition, the mechanism presented below is for the special case of a single private good i.e., $ K=1 $. A natural extension of the presented mechanism to the general case is discussed at the end of this section.
	
For any agent $ i\in \mN $, the message space is $ \mathcal{M}_i = \mathbb{R}^{N+1} $. The message $ m_i = (y_i,q_i) $ consists of agent $ i $'s demand $ y_i \in \mathbb{R} $ for the allocation of the single good and a surrogate/proxy $ q_i = \left( q_i^1,\ldots,q_i^N \right) \in \mathbb{R}^N $ for the demand of all the agents (including himself/herself). 
	
\begin{figure}
	\centering 
	\begin{tikzpicture}
		\node[circle,draw,line width=1pt] (i) at (0,0) {$ i $} ;
		\node[circle,draw] (nij) at (1,1) {$ a $} ;
		\node[above] (nijtext) at (-2,1.3) {$ n(i,j) = a $} ; 
		\node[above] (dijtext) at (-2,0.75) {$ d(i,j) = 3 $} ; 	
		\node[circle,draw] (b) at (2.5,1.5) {$ b $} ;
		\node[circle,draw,line width=1pt] (j) at (4,2) {$ j $} ;
		\node[circle,draw] (p) at (1.5,-1) {$ p $} ;
		\node[circle,draw] (q) at (-1,-1) {$ q $} ;
		\draw[-{Latex[width=1mm]}] (p) edge (i) ;
		\draw[-{Latex[width=1mm]}] (i) edge (q) ;
		\draw[-{Latex[width=1mm]}] (i) edge (nij) (nij) edge (b) (b) edge (j) ;
		\draw[dotted,line width=0.75pt] (p) -- (2.5,-1.5) ;		
		\draw[dotted,line width=0.75pt] (q) -- (-2,-1.5) ;
		\draw[dotted,line width=0.75pt] (nij) -- (2,0) ;
		\draw[dotted,line width=0.75pt] (b) -- (3,0.5) ;
		\draw[dotted,line width=0.75pt] (5,3) -- (j) -- (5,1) ;					
	\end{tikzpicture}
	\caption{$ n(i,j) $ and $ d(i,j) $ for the strongly connected directed graph $ \mathcal{G} $.}
	\label{figGraph} 
\end{figure}
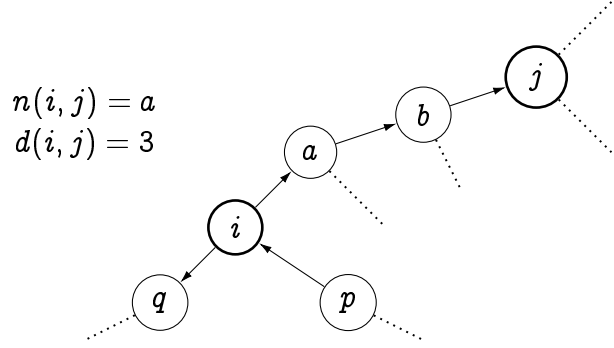
	
Since the underlying graph $ \mathcal{G} = \left( \mathcal{N},\mathcal{E} \right) $ is strongly connected, for any pair of vertexes $ i,j \in \mN $, the following two quantities are well-defined. $ d(i,j) $ is the length of the shortest path from $ i $ to $ j $ and $ n(i,j) \in \mN(i) $ is the out-neighbor of $ i $ such that the shortest path from $ i $ to $ j $ goes through $ n(i,j) $. The two quantities are depicted in Fig.~\ref{figGraph}. The allocation function is defined as
\begin{equation}\label{eqallopr}
	\widehat{x}_i(m) = y_i - \frac{1}{N-1} \sum_{r \in \mN(i)} \frac{q_{r}^r}{\xi} - \frac{1}{N-1} \sum_{\substack{r \notin \mN(i) \\ r \ne i} } \frac{q_{n(i,r)}^r}{\xi^{d(i,r)-1}} + \frac{c_1}{N}, \quad \forall~i \in \mN,
\end{equation}
where $ \xi \in (0,1) $ is an appropriately chosen \emph{contraction} parameter and its  selection is discussed in Section~\ref{seclearning} on Learning Guarantees, proof of Proposition~\ref{thmcontracpr}. The tax function is defined as
\begin{subequations}\label{eqtaxpr}
\begin{align}
	\label{eqtaxpr_a}
	\widehat{t}_i(m) &= \widehat{p}_i(m_{-i}) \left( \widehat{x}_i(m) - \frac{c_1}{N} \right) +  \left(  q_i^i - \xi y_i \right)^2 
	+ \sum_{\substack{r \in \mathcal{N}(i)}}  \left( q_i^r - \xi y_r \right)^2 
	+ \sum_{\substack{r \notin \mathcal{N}(i)\\r \ne i}}  \left(q_i^r - \xi q_{n(i,r)}^r \right)^2, 
	\\
	\label{eqtaxpr_b}
	\widehat{p}_i(m_{-i}) &= \frac{1}{\delta} \left( \frac{q_{n(i,i)}^i}{\xi} + \sum_{r \in \mN(i)} \frac{q_{r}^r}{\xi} + \sum_{\substack{r \notin \mN(i) \\ r \ne i} } \frac{q_{n(i,r)}^r}{\xi^{d(i,r)-1}} \right), \quad \forall~i \in \mN,  
\end{align} 
\end{subequations}
where $ n(i,i) \in \mN(i) $ is an arbitrarily chosen neighbor of $ i $ and $ \delta > 0 $ is an appropriately chosen parameter. Both $ \xi,\delta $ are selected simultaneously in the proof of Proposition~\ref{thmcontracpr} from Section~\ref{seclearning}. 
	
The quantities, $ n(\cdot,\cdot) $ and $ d(\cdot,\cdot) $, are based on the graph $ \mathcal{G} $. The only property of relevance here is that the two are related recursively i.e., $ d(n(i,r),r) = d(i,r) - 1 $. Thus if a mechanism designer wishes to avoid calculating the shortest path (possibly due to the high complexity) then $ n,d $ can be replaced by any valid neighbor and distance mapping, respectively, as long as they are related recursively as above.

\subsection{Results}

\begin{fact}[Distributed]
	The mechanism defined in~\eqref{eqallopr} and~\eqref{eqtaxpr} is distributed.
\end{fact}
The distributed-ness of the mechanism follows from the fact that the expressions in~\eqref{eqallopr} and~\eqref{eqtaxpr} depend only on $ m_i $ and $ \big(m_r\big)_{r \in \mN(i)} $.
	
Since agents are connected through a given graph, they can only communicate with a restricted set of agents i.e., their neighbors. Yet, as indicated in the optimality conditions,~\eqref{eqKKTpr}, there needs to be two kinds of global consensus at the efficient allocation. Firstly, the total allocation to all the agents must equal the total available amount, $ \sum_{i \in \mN} x_i^* = c_1 $. Secondly, agents also need to agree on a common ``price'', $ \lambda_1^* $. To facilitate this, the message space consists of the surrogate variables $ q_i = (q_i^1,\ldots,q_i^N) $ which are known locally to agent $ i $ and are expected at equilibrium to be representative of the global demand $ y = (y_1,\ldots,y_N) $. Specifically, the second, third and fourth terms in the tax,~\eqref{eqtaxpr_a}, are designed for incentivizing agents to achieve the aforementioned duplication of global demand $ y $ to the local surrogate $ q_i $. 
	
To motivate the choice of the allocation function and the remaining part of the tax function consider the case of $ \xi = 1 $ and take into account the duplication, i.e., $ q_r = y $, $ \forall \: r \in \mN $. Since $ \xi = 1 $, all the factors involving $ \xi $ become $ 1 $ for the expressions in~\eqref{eqallopr} and~\eqref{eqtaxpr}. We design the allocation $ \widehat{x}_i(m) $ as a function of $ y_i $ and $ \left( q_r \right)_{r \in \mN(i)} $ such that after taking into account the duplication it becomes $ y_i - \frac{1}{N-1} \sum_{j \ne i} y_j + \frac{c_1}{N} $. This facilitates the first global consensus, $ \sum_{i \in \mN} x_i = c_1 $. One standard design principle for mechanisms is that if an agent partially controls their own allocation (such as here, since $ \widehat{x}_i(m) $ depends on $ y_i $) then they shouldn't be able to control the price. This is the reason that $ \widehat{p}_i(\cdot) $ doesn't depend on $ m_i $. It is function of $ \left( q_r \right)_{r \in \mN(i)} $ and is designed such that after taking into account the duplication, the price for any agent is proportional to $ \sum_j y_j $. This facilitates the second consensus - common price for all agents.

Finally, we set $ \xi < 1 $ and adjust the allocation and tax function accordingly so that the game $ \mathfrak{G}_{pvt} $ can be contractive (see Section~\ref{seclearning}). 
	
Before proceeding to the main results, define the best-response of any agent $ i $,
\begin{equation}\label{eqBRdefpr}
	\beta_i(m_{-i}) = \big( {\breve{y}}_i(m_{-i}),\breve{q}_i(m_{-i}) \big) \triangleq  \argmax_{m_i \in \mathcal{M}_i} \, u_i(m).
\end{equation}
Denote $ \beta = \left(\beta_1,\ldots,\beta_N\right) $. Best-response $ \beta: \mathcal{M} \rightarrow \mathcal{M} $ is a set-valued function in general.
	
\begin{lemma}[Concavity]\label{thmconcpr}
	For any agent $ i \in \mN $ and $ m_{-i} \in \mathcal{M}_{-i} $, the utility $ u_i(m) $, defined in~\eqref{equtil}, for the game $ \mathfrak{G}_{pvt} $ is strictly concave in $ m_i=(y_i,q_i) $. Thus, the best-response of agent $ i $ is unique and is defined by the first order conditions.	
\end{lemma}
\begin{proof}
	Please see Appendix~\ref{proofconcpr}.
\end{proof}
	
Concavity of the induced utility in the game $ \mathfrak{G}_{pvt} $ largely follows from the tax terms being quadratic. The second tax term in~\eqref{eqtaxpr_a} is the only source of cross derivatives across components of message $ m_i $. Concavity is proved by verifying that the Hessian matrix is negative definite.

\begin{theorem}[Full Implementation and Budget Balance] \label{thmFIpr}
	For the game $ \mathfrak{G}_{pvt} $, there exists a unique Nash equilibrium, $ \widetilde{m} \in \mathcal{M} $, and the allocation at Nash equilibrium is efficient, i.e., $ \widehat{x}_i(\widetilde{m}) = x_i^* $, $ \forall $ $ i \in \mN $. Further, the total tax paid at Nash equilibrium $ \widetilde{m} $ is zero, i.e.,
\begin{equation}
	\sum_{i\in\mN} \widehat{t}_i(\widetilde{m}) = 0.
\end{equation}
\end{theorem}
\begin{proof}
	Please see Appendix~\ref{proofFIpr}.
\end{proof}
	
The proof of Proposition~\ref{thmFIpr} can be intuitively explained as follows. Since the optimality conditions in~\eqref{eqKKTpr} are sufficient, we start by showing that at any Nash equilibrium the allocation and price necessarily satisfy the optimality conditions. Thereby ensuring that if Nash equilibrium exists (unique or multiple) the corresponding allocation is efficient. Then we show existence and uniqueness by showing a one-to-one map between message at Nash equilibrium and $ (x^*,\lambda^*) $ arising out of the optimization in~\eqref{eqcppr}.

\subsubsection*{Generalizing to multiple goods $ (K > 1) $}
	
For the general problem we use notation on a per good basis. The message space is, 
\begin{equation}
	\mathcal{M}_i = \underset{k \in \mathcal{K}}{\times} \mathcal{M}_i^k = \underset{k \in \mathcal{K}}{\times} \mathbb{R}^{(N+1)} = \mathbb{R}^{(N+1)K}.
\end{equation}
Any message $ m_i = \left( m_i^k \right)_{k \in \mathcal{K}} =  (y_i,q_i) $  contains separate demands and proxies for each good $ k \in \mathcal{K} $. Denote it as follows
\begin{subequations}
\begin{align}
	y_i &= \left(y_i^k\right)_{k \in \mathcal{K}} \in \mathbb{R}^K,
	\\
	q_i &= \big( q_i^{r,k} \big)_{r \in \mN, k \in \mathcal{K}} \in \mathbb{R}^{NK},
\end{align}
\end{subequations}
where for any good $ k \in \mathcal{K} $, demand and proxies $ \big( y_i^k, (q_i^{r,k} )_{r \in \mN} \big) \in \mathbb{R}^{N+1} $ have the same interpretation as $ (y_i,q_i) $ in the presented mechanism above with only one good. The allocation function is $ \widehat{x}_i^k(\cdot) $, for any $ i \in \mN $, $ k \in \mathcal{K} $ and the tax function is
\begin{equation}
	\widehat{t}_i(\cdot) = \sum_{k \in \mathcal{K}} \widehat{t}_i^{\, k}(m), \quad\forall~i \in \mN.
\end{equation}
Here both functions $ \widehat{x}_i^k(\cdot) $ and $ \widehat{t}_i^{\, k}(\cdot) $, depend only on $ m^k $, the part of message $ m $ pertaining to good $ k $. The expression for both functions are the same as in the presented mechanism above,~\eqref{eqallopr} and~\eqref{eqtaxpr}, replacing $ m $ by $ m^k $.

The generalized mechanism is distributed as well, since all allocation and tax functions still depend only on message of the agent and that of its neighbors. Owing to the design with a good-wise separation, the utility has the following form
\begin{equation}
	u_i(m) = v_i\left( \left( \widehat{x}_i^k(m^k) \right)_{k \in \mathcal{K}} \right) - \sum_{k \in \mathcal{K}}\widehat{t}_i^{\,k}(m^k).
\end{equation}
Concavity of $ u_i $ follows from arguments similar to Proposition~\ref{thmconcpr}, where we verify the Hessian to be negative definite. Since the optimality conditions in~\eqref{eqKKTpr} are sufficient, the properties of efficiency, existence and uniqueness of Nash equilibrium follow from the first order conditions for optimality in the best-response. Finally, the Budget Balance result holds true on a per-good basis, so it also holds for the total tax which is the sum of per-good taxes.

\section{A mechanism for the public goods problem}
\label{secmechpub}
	
For the purpose of maintaining clarity in the exposition, the presented mechanism below is for the special case of a single feature in the public good i.e., $ P=1 $. A natural extension to the general case is discussed at the end of this section.
	
For any agent $ i\in \mN $, the message space is $ \mathcal{M}_i = \mathbb{R}^{N+1} $. The message $ m_i = (y_i,q_i) $ consists of agent $ i $'s contribution $ y_i \in \mathbb{R} $ to the common public good and a surrogate/proxy $ q_i = \left( q_i^1,\ldots,q_i^N \right) \in \mathbb{R}^N $ for the contributions of all the agents (including himself/herself). 
	
The allocation function is defined as
\begin{equation}\label{eqallopub}
	\widehat{x}_i(m) = \frac{1}{N} \left( y_i + \sum_{r \in \mN(i)} \frac{q_{r}^r}{\xi} + \sum_{\substack{r \notin \mN(i) \\ r \ne i} } \frac{q_{n(i,r)}^r}{\xi^{d(i,r)-1}}\right), \quad \forall~i \in \mN.
\end{equation}
The tax function is 
\begin{subequations}\label{eqtaxpub}
\begin{align}
	\nonumber 
	\widehat{t}_i(m) &= \widehat{p}_i(m_{-i}) \widehat{x}_i(m) +  \left(  q_i^i - \xi y_i \right)^2
	+ \sum_{\substack{r \in \mathcal{N}(i)}}  \left( q_i^r - \xi y_r \right)^2 
	+ \sum_{\substack{r \notin \mathcal{N}(i)\\r \ne i}}  \left(q_i^r - \xi q_{n(i,r)}^r \right)^2,
	\\ 	\label{eqtaxpub_a}
	&\quad + \frac{\delta}{2} \left( q_{n(i,i)}^i - \xi y_i \right)^2  
	\\
	\label{eqtaxpub_b}
	\widehat{p}_i(m_{-i}) &= \delta(N-1) \left( \frac{q_{n(i,i)}^i}{\xi} - \frac{1}{N-1} \sum_{r \in \mN(i)} \frac{q_{r}^r}{\xi} - \frac{1}{N-1}  \sum_{\substack{r \notin \mN(i) \\ r \ne i} } \frac{q_{n(i,r)}^r}{\xi^{d(i,r)-1}}  \right), \quad \forall~i \in \mN,  
\end{align} 
\end{subequations}
where $ n(i,i) \in \mN(i) $ is an arbitrarily chosen neighbor of $ i $ and $ \xi \in (0,1) $, $ \delta > 0 $ are appropriately chosen parameters and their selection is discussed in Section~\ref{seclearning} on Learning Guarantees, proof of Proposition~\ref{thmcontracpub}.
	
In the following, denote the induced game by $ \mathfrak{G}_{pub} $ and define the utility in the game,~\eqref{equtil}, and the best-response,~\eqref{eqBRdefpr}, as in the previous section.

\subsection{Results} 
	
\begin{fact}[Distributed]
	The mechanism defined in~\eqref{eqallopub} and~\eqref{eqtaxpub} is distributed.
\end{fact}
	
The optimality conditions in~\eqref{eqKKTpub} require that agents have global consensus on two aspects: allocation must be the same for all and sum of prices should be equal to zero. As in the private goods mechanism, we design the second, third and fourth tax terms in~\eqref{eqtaxpub_a} such that agents are incentivized to duplicate global demand $ y $ onto locally available variables $ q_i $. To motivate the allocation function and the remaining part of the tax function consider the case of $ \xi = 1 $ and take into account the duplication $ q_r = y $, $ \forall~r\in \mN $. In this case all the factors involving $ \xi $ in~\eqref{eqallopub} and~\eqref{eqtaxpub} are 1. The allocation function $ \widehat{x}_i(m) $ depends on $ y_i, \left( q_r \right)_{r \in \mN(i)} $ and is designed such that after taking into account the duplication it is proportional to $ \sum_j y_j $. This facilitates the first consensus - all agents' allocation must be the same. The price $ \widehat{p}_i(m_{-i}) $ is designed such that it depends only on $ \left( q_r \right)_{r \in \mN(i)} $ and after taking into account the duplication it is proportional to $ y_i - \frac{1}{N-1} \sum_{j \ne i} y_j $. This facilitates the second consensus - sum of prices over all agents is zero. 
	
With the above design principles, all the results of this section follow. We then introduce an additional fifth term in the tax,~\eqref{eqtaxpub_a}, just for the purpose of achieving contraction in Section~\ref{seclearning} (see proof of Proposition~\ref{thmcontracpub}). Incentives provided by this term are in line with those already provided to the neighboring agent $ n(i,i) $ through his/her third tax term, hence it doesn't interfere with the equilibrium results in this section. Only the proof of concavity changes slightly.
		
Finally, we set $ \xi < 1 $ and adjust everything in the allocation and tax function correspondingly so that the game $ \mathfrak{G}_{pub} $ can be contractive (see Section~\ref{seclearning}).
	
For the next two results, the basic idea behind the proofs are similar to the corresponding results from the previous section.	
	
\begin{lemma}[Concavity]\label{thmconcpub}
	For any agent $ i \in \mN $ and $ m_{-i} \in \mathcal{M}_{-i} $, the utility $ u_i(m) $ for the game $ \mathfrak{G}_{pub} $ is strictly concave in $ m_i=(y_i,q_i) $. Thus, the best-response of agent $ i $ is unique and is defined by the first order conditions.
\end{lemma}
\begin{proof}
	Please see Appendix~\ref{proofconcpub}.	
\end{proof}

\begin{theorem}[Full Implementation and Budget Balance] \label{thmFIpub}
	For the game $ \mathfrak{G}_{pub} $, there exists a unique Nash equilibrium, $ \widetilde{m} \in \mathcal{M} $, and the allocation at Nash equilibrium is efficient, i.e., $ \widehat{x}(\widetilde{m}) = x^* $. Further, the total tax paid at Nash equilibrium $ \widetilde{m} $ is zero, i.e., 
\begin{equation}
	\sum_{i\in\mN} \widehat{t}_i(\widetilde{m}) = 0.
\end{equation}
\end{theorem}
\begin{proof}
	Please see Appendix~\ref{proofFIpub}.	
\end{proof}

\subsubsection*{Generalizing to multiple features $ (P > 1) $}
	
With an idea similar to the extension for the private goods mechanism, extend the presented mechanism by first increasing the message space such that for each agent $ m_i = \big( m_i^p \big)_{p \in \mathcal{P}} = (y_i,q_i) $. The allocation in this case is $ P- $dimensional and the expression for $ \widehat{x}_i^p(m) $ is the same as in the presented mechanism with $ y_i,\big(q_r\big)_{r \in \mN(i)} $ replaced by $ y_i^p,\big(q_r^p\big)_{r \in \mN(i)} $. The tax function is
\begin{equation}
	\widehat{t}_i(m) = \sum_{p \in \mathcal{P}} \widehat{t}_i^{\: p}(m),
\end{equation}
where the expression for $ \widehat{t}_i^{\: p} $ is the same as in the presented mechanism, replacing $ m $ by $ m^k $.
	
The results above follow using analogous arguments to those mentioned in the previous section regarding generalization to multiple goods, $ K > 1 $.

\section{Learning Guarantees} \label{seclearning} 
	
This section provides the result for guaranteed convergence for a class of learning algorithms, when the mechanisms defined in Sections~\ref{secpr} and~\ref{secmechpub} are played repeatedly. As discussed in the Introduction, such results act as a measure of robustness of a mechanism, w.r.t. information available to agents, and thus makes the mechanism ready for practical applications.  
	
A learning dynamic is represented by 
\begin{equation}
	\left( \mu_n \right)_{n \ge 1} \subseteq \underset{i \in \mN}{\times} \Delta(\mathcal{M}_i),
\end{equation}
where $ \mu_n $ is a mixed strategy profile with product structure to be used at time $ n $. Denote by $ S(\mu_n) \subseteq \mathcal{M} $ the support of the mixed strategy profile $ \mu_n $ and denote by $ m_n \in S(\mu_n) $ the realized action. Healy and Mathevet in~\cite{healy2012designing} define the Adaptive Best-Response (ABR) dynamics class by restricting the support $ S(\mu_n) $ in terms of past observed actions. Define the history $ H_{n^\prime,n} = \left( m_{n^\prime},m_{n^\prime + 1},\ldots,m_{n-1} \right) $ as the set of observed actions between $ n^\prime $ and $ n-1 $. Denote by $ \widetilde{m} $ the unique Nash equilibrium of the game and define $ B\left( \mathcal{M}^\prime \right) $ as the smallest closed ball centered at $ \widetilde{m} $ that contains the set $ \mathcal{M}^\prime \subset \mathcal{M} $. The closed ball is defined with any valid metric $ d $ on the message space $ \mathcal{M} $. 
	
A learning dynamic is in the ABR class if any point in the support of the action at time $ n $ is no further from the Nash equilibrium than the best-response to any action that is no further from Nash equilibrium than the ``worst-case'' action that has been observed in some finite past $ \{ n^\prime,\ldots,n-1 \} $. 
	
\begin{definition}[Adaptive Best-Response Learning Class~{\cite{healy2012designing}}]
	A learning dynamic is an adaptive best-response dynamic if $ \forall~n^\prime,~\exists~\hat{n} > n^\prime,~\text{s.t.}~\forall~ n \ge \hat{n} $,
\begin{equation}
	S(\mu_n) \subseteq B\left( \beta\left( B\left( H_{n^\prime,n} \right) \right) \right), 
\end{equation}
	where $ \beta:\mathcal{M} \rightarrow \mathcal{M} $ is the best-response of the game.
\end{definition}
	
The above is satisfied for instance if every agent puts belief zero over actions further from Nash equilibrium than the ones that he/she has observed in the past. Some well-known learning dynamics are in the ABR class, following are a few examples.
	
Cournot best-response is defined as $ S(\mu_n) = m_n = \beta(m_{n-1}) $, i.e., at every time agents best-respond to the last round's action. This gives rise to a deterministic strategy at each time. More generally, $ k- $period best-response is defined as the learning dynamic where at any time $ n $, an agent $ i $'s strategy is a best-response to the mixed strategy of agents $ j \ne i $ which are created using the observed empirical distribution from the actions of the previous $ k- $rounds i.e., $ \{ m_{j,n-k},\ldots,m_{j,n-1} \} $. In fact, the generalization of this is also in the ABR class, where at each time $ n $, an agent $ i $'s strategy is the best-response to the mixed strategy of agents $ j \ne i $ that is formed by taking any convex combination of the empirical distributions of actions observed in the previous $ k- $rounds. Finally, Fictitious Play~\cite{brown,fudenberglearning}, which maintains empirical distribution of all the past actions (instead of $ k $ most recent ones) is also in ABR. The additional requirement for this is that the utility in the game should be strictly concave, which is true for the presented mechanisms (see Lemmas~\ref{thmconcpr} and~\ref{thmconcpub}).

\begin{definition}[Contractive Mechanism]
	Let $ d $ be any metric defined on the message space $ \mathcal{M} $ such that $ \left(\mathcal{M},d\right) $ is a complete metric space. A mechanism is called contractive if for any profile of utility function $ \big( v_i(\cdot) \big)_{i \in \mN} $ that satisfy the assumptions of the model, the induced game $ \mathfrak{G}_{pvt} $ or $ \mathfrak{G}_{pub} $ (depending on the allocation problem) has a single-valued best-response function $ \beta: \mathcal{M} \rightarrow \mathcal{M} $ that is a $ d- $contraction mapping.
\end{definition}
	
A function $ h:\mathcal{M} \rightarrow \mathcal{M} $ is a $ d- $contraction mapping if $ \Vert h(x) - h(y) \Vert_d < \Vert x - y \Vert_d  $ for all distinct $ x,y \in \mathcal{M} $. For the results below, the following well-known check for contraction mapping is used: $ h $ is a contraction mapping (for some metric $ d $) if the Jacobian has norm less than one, i.e., $ \Vert \nabla h \Vert < 1 $, where any matrix norm can be considered. Specifically, we consider the \emph{row-sum} norm. 
	
For the game induced by a contractive mechanism, by definition, there is a unique Nash equilibrium. This is due to the Banach fixed-point theorem, which gives that the best-response iteration for the induced game converges to a unique point. As shown in the previous sections, the game $ \mathfrak{G}_{pvt} $ and $ \mathfrak{G}_{pub} $ already have a unique Nash equilibrium.

\begin{fact}[{\cite[Theorem 1]{healy2012designing}}] \label{factabr}
	If a game is contractive, then all ABR dynamics converge to the unique Nash equilibrium.
\end{fact}
	
The idea behind the above proof is to show that after a finite time, under any ABR dynamic the distance to equilibrium gets smaller (exponentially) between successively rounds due to the best-response being a contraction mapping. It is shown below that the presented mechanism for both, the private and public goods problems, is contractive and thus owing to the above result there is guaranteed convergence for all learning dynamics in the ABR class.
	
Contraction ensures convergence for the ABR class; this result is in the same vein as the one in the seminal work~\cite{milgrom1990}. Milgrom and Roberts show that Supermodularity ensures convergence for the Adaptive Dynamics class of learning algorithms (also defined in~\cite{milgrom1990}). Supermodularity requires that the best-response of any agent $ i $ is non-decreasing in the message $ m_j $ of any other agent $ j \ne i $. The aim in this paper is to get guarantees for the ABR class, however it is shown below that the game $ \mathfrak{G}_{pvt} $ induced by the private goods mechanism is also supermodular and thus has guaranteed convergence for the Adaptive Dynamics class as well.
	
Before proceeding, kindly note that contraction is somewhat more stringent a condition that supermodularity and consequently the ABR class is broader than the adaptive dynamics class. Thus in order to have better learning guarantees, our principle aim to ensure the property of contraction.

\subsection{Results}

\begin{theorem}[Contraction]\label{thmcontracpr}
	The game $ \mathfrak{G}_{pvt} $ defined in Section~\ref{secmechpr} is contractive but not supermodular. Thus, all learning dynamics within the ABR dynamics class converge to the unique efficient Nash equilibrium.
\end{theorem}
\begin{proof}
	Please see Appendix~\ref{proofcontracpr}.
\end{proof}
	
The intuition behind the proof of Proposition~\ref{thmcontracpr} can be motivated as follows. By selecting parameter $ \xi < 1 $, contraction of best-response for the variables $ \left( \tilde{q}_i^r \right)_{r \ne i} $ is already ensured. However, for best-response in $ \tilde{y}_i,\tilde{q}_i^i $, tuning of $ \xi,\delta $ is needed. Due to the specific nature of the mechanism it turns out that in order to accommodate any given value of $ \eta > 1 $, $ \xi $ needs to be selected close to $ 1 $ and correspondingly $ \delta $ needs to be selected as a function of the chosen $ \xi $. Finally, it is also shown that the best-response $ \tilde{y}_i $ and $ \tilde{q}_i^i $ are decreasing in $ q_{n(i,i)}^i $ and hence the game is not supermodular.

\paragraph*{Generalizing to $ K>1 $} The proof of Proposition~\ref{thmcontracpr} relies on inverting $ v_i^\prime:\mathbb{R} \rightarrow \mathbb{R} $ and bounding it appropriately. For the general case this is the same as inverting $ \nabla v_i:\mathbb{R}^K \rightarrow \mathbb{R}^K $. Strict concavity of $ v_i(\cdot) $ ensures that the determinant of Hessian of $ v_i $ is never zero, hence an inverse for $ \nabla v_i $ exists in the general case. Of course, to bound the derivatives of $ ( \nabla v_i )^{-1} $, instead of using the condition in~\eqref{eqetapr} we use the more general condition from~\eqref{eqetagenpr}. Finally, the proof is completed by tuning the parameters $ \xi,\delta $ in exactly the same manner as in the proof of Proposition~\ref{thmcontracpr}.
	

\begin{theorem}[Contraction]\label{thmcontracpub}
	The game $ \mathfrak{G}_{pub} $ defined in Section~\ref{secmechpub} is both contractive and supermodular. Thus, all learning dynamics within the ABR dynamics class converge to the unique Nash equilibrium.
\end{theorem}
\begin{proof}
	Please see Appendix~\ref{proofcontracpub}.
\end{proof}
	
Initially the parameters have the following bound: $ \xi \in (0,1) $ and $ \delta > 0 $. In order to get contraction, a further restriction $ \xi \in \left( \sqrt{\left. (N-1) \middle/ N \right.},1 \right) $ needs to be imposed. This also gives that each best-response is non-decreasing in the message of every other agent and thus the game is supermodular. Finally, here too in order to accommodate any value of $ \eta > 1 $ the final tuning of $ \xi $ requires it to be chosen very close to $ 1 $ and consequently $ \delta $ is selected as a function of the chosen $ \xi $.
	
Generalizing to $ P > 1 $, the essential idea of inverting $ \nabla v_i $ from above works here too and from there onwards the steps follow analogously to the ones in the proof above.
	
\subsection{Numerical Analysis of convergence}

For numerical analysis we consider the private goods problem with one good $ (K=1) $ and the public goods problem with one feature $ (P=1) $ and $ N = 31 $, $ \eta = 25 $. The agents' utility function as quadratic, i.e.,
\begin{align}
	v_i(x_i) &= \theta_i x_i^2 + \sigma_i x_i, 
	\tag{Private}
	\\
	v_i(x) &= \theta_i x^2 + \sigma_i x.
	\tag{Public} 
\end{align}
An example of quadratic utility function can be found in~\cite{samadi2012}, for the model of demand side management in smart-grids. In each case the second derivative of $ v_i(\cdot) $ is $ 2\theta_i $ and thus for any agent $ i $ the value for $ \theta_i $ is chosen uniformly randomly in the range $ (-\frac{\eta}{2},-\frac{1}{2\eta}) $. As the model doesn't impose any restriction on the first derivative, the value for $ \sigma_i $ is chosen uniformly randomly in the range $ (10,20) $. From the proof of Propositions~\ref{thmcontracpr} and~\ref{thmcontracpub}, one can numerically calculate the value of parameters $ \xi,\delta $. For the particular instance of the random $ \theta,\sigma $ generated to be used for the plots below (the same values for $ \theta,\sigma $ are used for both public and private goods examples), the parameter values are listed in Table~\ref{table1}. The two cases for graph $ \mathcal{G} $ considered are: a full binary tree and a sample of the Erd\H{o}s-Re\'{n}yi random graph with only one connected component, where any two edges are connected with probability $ p = 0.3 $. The first represents a case of small average degree whereas the second represents the case of large average degree. The same instance of Erd\H{o}s-Re\'{n}yi random graph is used for both private and public goods examples.

\begin{table} 
	\centering 
	\begin{tabular}{|c|l|l|} \hline
			& \multicolumn{1}{|c|}{Private goods} & \multicolumn{1}{|c|}{Public goods} 
			\\ \hline 
			Binary Tree & \begin{tabular}{l}
				$\xi = 1 - (1.831\times 10^{-4}) $, \\ 
				$ \delta = 1005.6 $.
			\end{tabular} & \begin{tabular}{l}
				$ \xi = 1 - (2.515 \times 10^{-4}) $, \\ 
				$ \delta = 0.9505 $.
			\end{tabular} \\ \hline
			Erd\H{o}s-Re\'{n}yi & \begin{tabular}{l}
				$\xi = 1 - (7.324 \times 10^{-4}) $, \\ 
				$ \delta = 932.3 $.
			\end{tabular} & \begin{tabular}{l}
				$ \xi = 1 - (10^{-3}) $, \\ 
				$ \delta = 0.8744 $.
			\end{tabular} \\ \hline  
	\end{tabular}
	\caption{$ (\xi,\delta) $ parameter values for different graphs and problems.}
	\label{table1}
\end{table}

Since it has been shown that the best-response is a contraction mapping, one expects that any learning strategy that best-responds to some convex combination of past actions from finitely many rounds, converges at an exponential rate. Indeed, this is exactly observed from Fig.~\ref{figplotpr} and~\ref{figplotpub}, where the absolute distance $ \Vert m_n - \widetilde{m} \Vert_2 $ of the action in round $ n $ to the Nash equilibrium $ \widetilde{m} $ of the game $ \mathfrak{G}_{pvt} $ and $ \mathfrak{G}_{pub} $ is plotted versus $ n $, respectively. For the learning dynamic\footnote{A theoretical simplification with quadratic utilities $ v_i(\cdot) $ is that for learning strategies such as $ k- $period best-response or Fictitious Play, instead of maintaining the empirical distribution over other agents' past actions, every agent can simply maintain the empirical average.}, one case considered is when the action taken by any agent $ i $ at time $ n $ is the best-response to an exponentially weighed average of past actions, i.e.,
\begin{subequations}
\begin{align}
	m_{n} &= \beta\left( \frac{m_{n-1}}{2}  + \frac{r_{n-1}}{2}  \right),
	\\
	r_{n} &= \frac{m_n}{2} + \frac{r_{n-1}}{2}.
	\end{align}   
\end{subequations} 
The second learning dynamic is to best-respond to the average of past 10 rounds.

\begin{figure}
	\centering 
	\includegraphics[width=\textwidth]{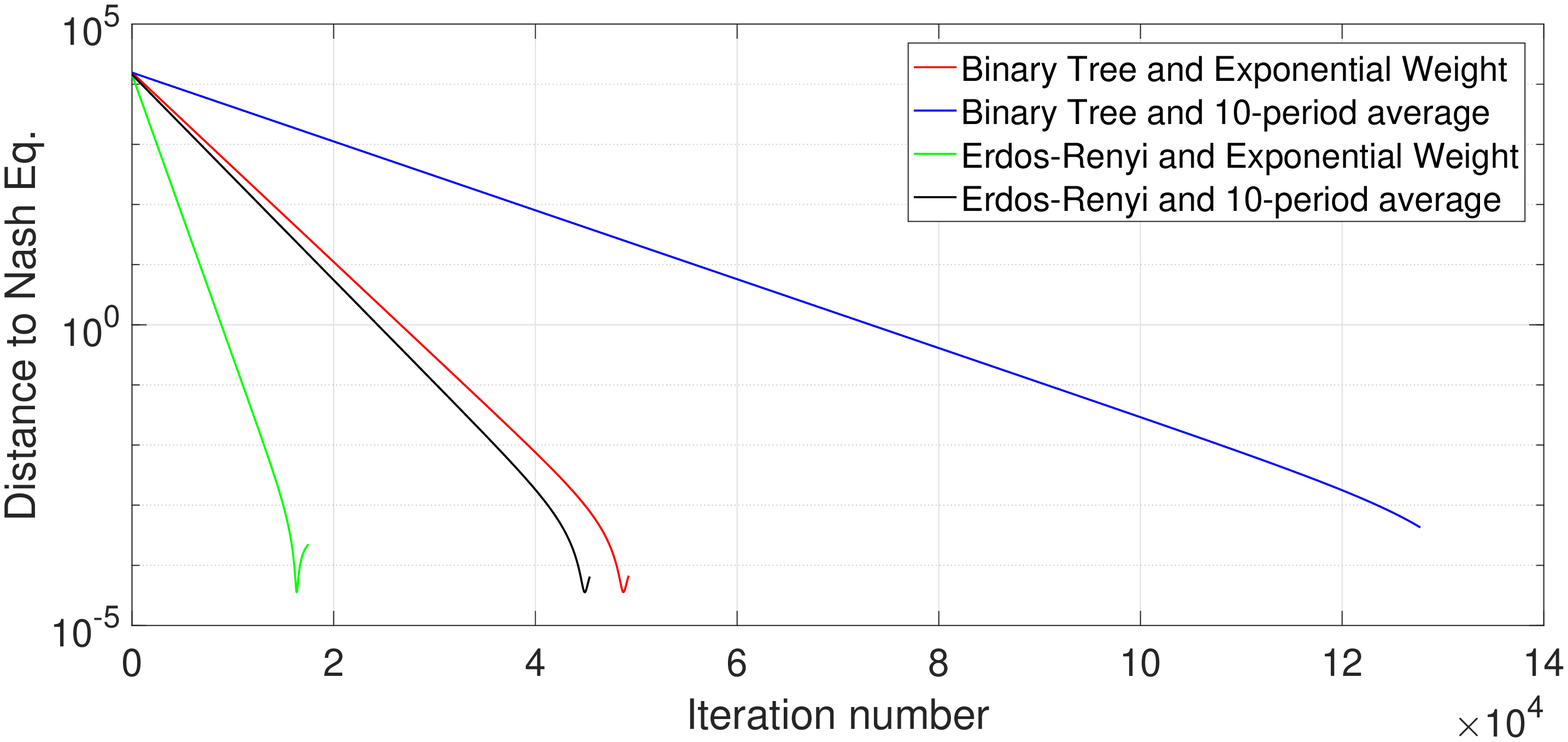}
	\caption{$ \Vert m_n - \widetilde{m} \Vert_2 $ vs. $ n $ for the private goods mechanism.}
	\label{figplotpr}
\end{figure}
	
\begin{figure}
	\centering 
	\includegraphics[width=\textwidth]{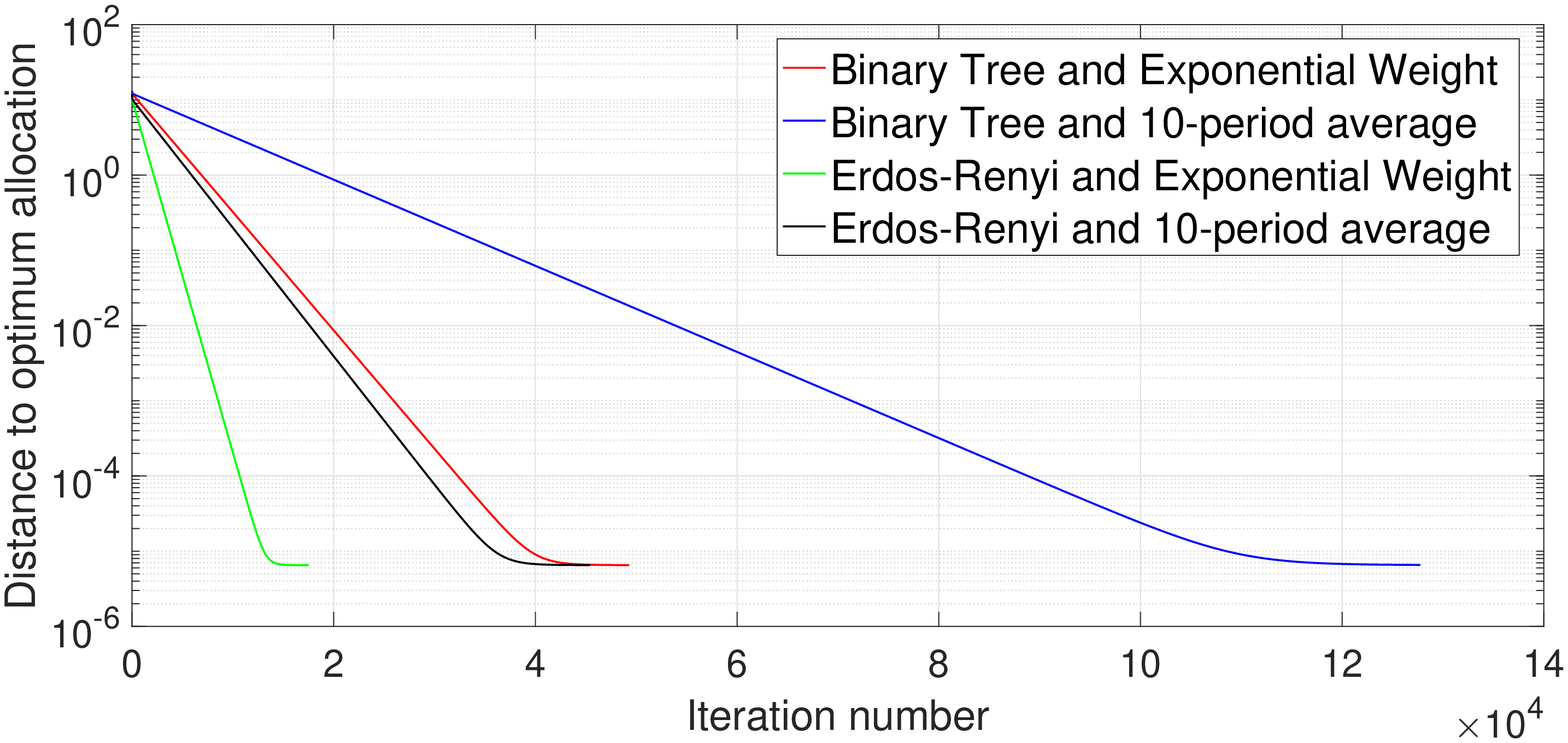}
	\caption{$ \Vert \widehat{x}(m_n) - x^* \Vert_2 $ vs. $ n $ for the private goods mechanism.}
	\label{figplotpr_xdist}	
\end{figure}
	
\begin{figure}
	\centering 
	\includegraphics[width=\textwidth]{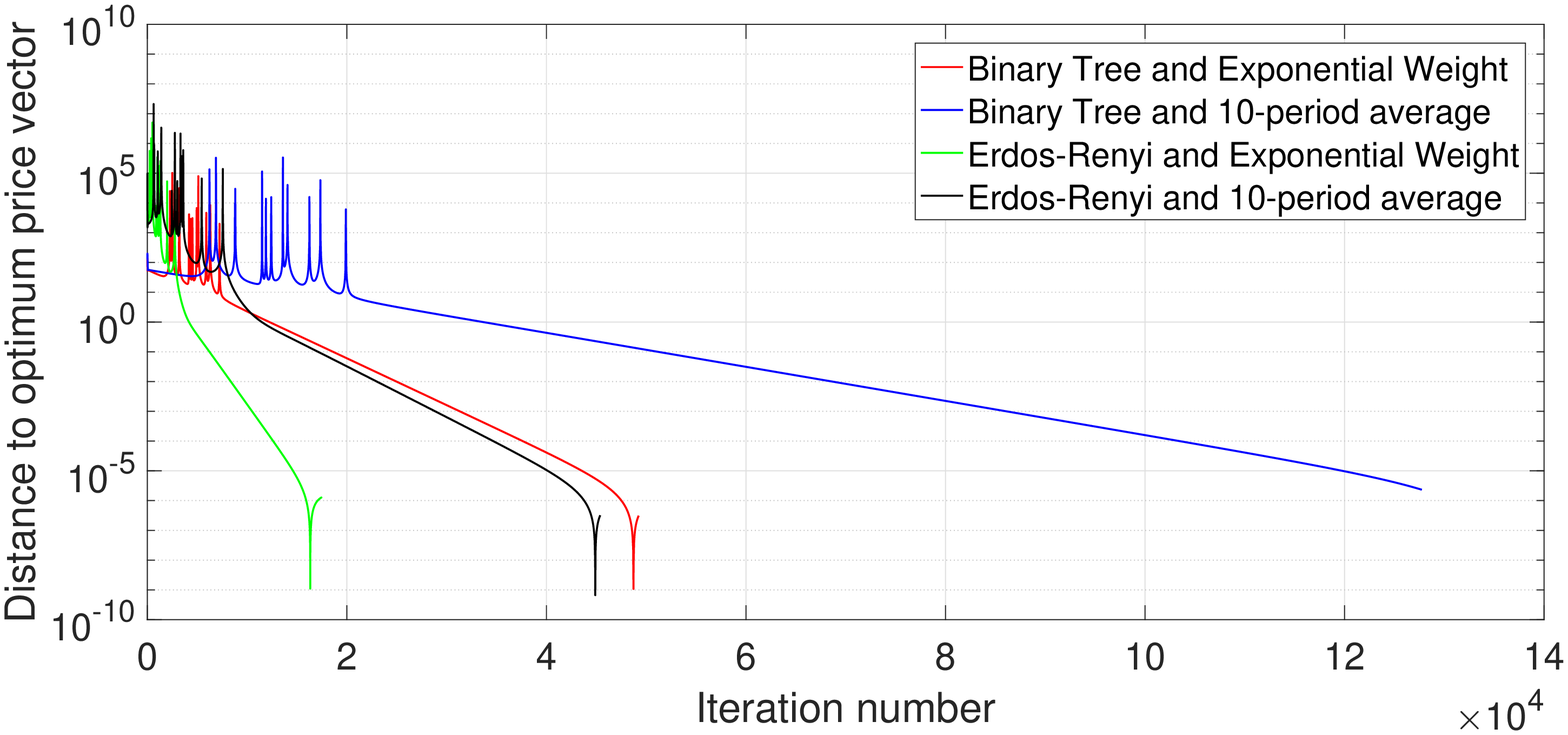}
	\caption{$ \Vert \left( \widehat{p}_i(m_n) \right)_{i \in \mN} - (\lambda_1^*,\ldots,\lambda_1^*) \Vert_2 $ vs. $ n $ for the private goods mechanism.}
	\label{figplotpr_pdist}
\end{figure}
	
\begin{figure}
	\centering 
	\includegraphics[width=\textwidth]{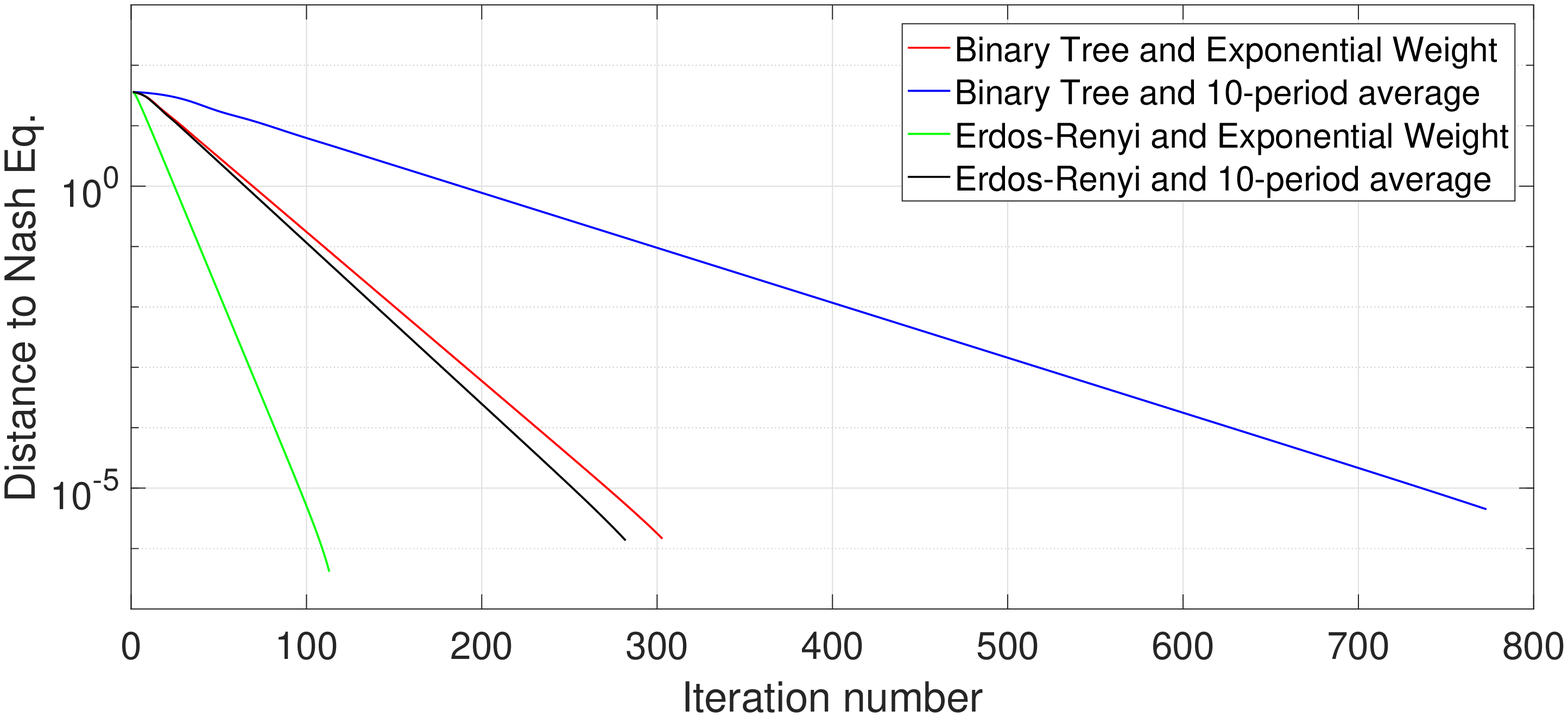}
	\caption{$ \Vert m_n - \widetilde{m} \Vert_2 $ vs. $ n $ for the public goods mechanism.}
	\label{figplotpub}
\end{figure}

Since the learning iterations are essentially conducting information exchange, one expects that the convergence to be faster for a graph that is more connected. From both the figures this is evident, as for each learning dynamic higher average degree Erd\H{o}s-Re\'{n}yi random graph shows faster convergence than lower average degree Binary Tree. In fact, for both learning dynamics the convergence for the Erd\H{o}s-Re\'{n}yi random graph is faster than either learning dynamic for Binary Tree. Comparing the two learning dynamics among themselves, we observe that the more aggressive exponential weighing leads to faster convergence compared to the learning dynamic that puts equal weight on each of the previous 10 actions. Finally, for Fig.~\ref{figplotpr}, in each case the relative distance to Nash equilibrium, defined as $ \left. \Vert m_n - \widetilde{m} \Vert_2 \middle/ \Vert \widetilde{m} \Vert_1 \right. $, is in the order of $ 10^{-9} $ when the absolute distance to Nash equilibrium is $ 10^{-3} $. For Fig.~\ref{figplotpub}, the relative distance in each case is of the order of $ 10^{-8} $ when the absolute distance to Nash equilibrium is $ 10^{-5} $.

Fig.~\ref{figplotpr_xdist} and~\ref{figplotpr_pdist} refer to the private goods mechanism and plot the distance of allocation $ \widehat{x}(m_n) $ and price $ \left( p_1(m_n),\ldots,p_N(m_n) \right) $ in round $ n $ to $ x^* $ and $ \left( \lambda_1^*,\ldots,\lambda_1^* \right) $, respectively. The convergence pattern is the same as the one observed in Fig.~\ref{figplotpr}, for the distance of message $ m_n $ to Nash equilibrium $ \widetilde{m} $.

\section{Conclusion}

In this paper, we present a distributed mechanism where agents only need to exchange messages locally with their neighbors. While for models with non-strategic agents, extensive research has been done in the field of distributed learning and optimization, this is not the case with mechanism design where agents are fully strategic.  
	%
For every profile $ \left( v_i(\cdot) \right)_{i \in \mN} $ of utility functions, the induced game is shown to have a unique NE. The allocation at equilibrium is efficient and taxes are budget balanced. Then we establish informational robustness of the mechanism by showing that the best-response in the induced game is a contraction mapping. This establishes that every learning dynamic within the ABR dynamics class converges to the unique and efficient NE when the game is played repeatedly. The ABR class contains learning dynamics such as Cournot best-response, $ k- $period best-response and Fictitious Play. 
	
\subsubsection*{Future Work}
	
A significant scope for improvement in the presented mechanism is the reduction of the size of the message space. 
A more scalable mechanism would be one where on average each agent's message space is of dimension $ o(N) $. However, such an attempt might possibly require restrictions on either the underlying graph $ \mathcal{G} $ or the upper bound on $ \eta $ that is admissible under the model. For the presented mechanism, the only restriction on the graph is that it is connected and there is no upper bound on $ \eta $. 
Another improvement can be that of considering more complicated constraint sets for the optimization~\eqref{eqcppr}. However, as can be seen from previous attempts at mechanism design for general constraint sets,~\cite{johari2009efficiency,rahuljain,demos,SiAn_multicast_tcns,SiAn14b}, such an extension is not straightforward. 
	%
Finally, for preventing inter-temporal exchange of money during the learning phase, one can adjust taxes such that there is budget balance for all messages, rather than just at NE.

\bibliographystyle{unsrtnat}
{
	\bibliography{abhinav}
}

\appendix
	
\section{Proof of Lemma~\ref{thmconcpr} (Concavity - Private goods)}
\label{proofconcpr}

\begin{proof}
Since the allocation and tax functions are smooth and $ v_i(\cdot) $ is continuously double-differentiable, to establish concavity we show that the Hessian of $ u_i(m) $ w.r.t. $ m_i $ is negative definite i.e., $ H \prec 0 $. Once this is established, the optimization in~\eqref{eqBRdefpr} has a strictly concave objective and an unbounded constraint set. Thus it has a unique maximizer, defined by the first order derivative conditions.
		
The Hessian is of size $ (N+1) \times (N+1) $ and we have
\begin{subequations}	
\begin{align}
	H_{11} &= \frac{\partial^2 u_i(m)}{\partial y_i^2} = v_i^{\prime\prime}(\widehat{x}_i(m)) - 2\xi^2, \\
	H_{(j+1)1} = H_{1(j+1)} &= \frac{\partial^2 u_i(m)}{\partial y_i \partial q_i^j} =
	\left\{
	\begin{array}{ll}
	0  & \mbox{for } j \in \mN,~j\ne i, \\
	2\xi & \mbox{for } j=i,
	\end{array}
	\right. 
	\\
	H_{(j+1)(j+1)} &= \frac{\partial^2 u_i(m)}{\partial (q_i^j)^2} = -2, \quad \forall~ j \in \mN,\\	
	H_{(r+1)(j+1)} &= \frac{\partial^2 u_i(m)}{\partial q_i^r \partial q_i^j} = 0 , \quad \forall~ j,r \in \mN,~j \ne r.
\end{align}
\end{subequations}
The characteristic equation, $ \textsf{Det}\left(H - xI\right) = 0 $, becomes
\begin{equation}
	\left(x+2\right)^{N-1} \Big( (x+2)(x-H_{11}) - 4\xi^2 \Big) = 0.
\end{equation}
This implies that $ N-1 $ eigenvalues of $ H $ are $ -2 $ and the remaining two eigenvalues satisfy $ x^2 + (2-H_{11})x - 2v_i^{\prime\prime}(\widehat{x}_i(m)) = 0 $. Since $ H $ is a symmetric matrix, all its eigenvalues are real. Due to $ v_i^{\prime\prime}(\cdot) < 0 $, the product of roots in the above quadratic equation is positive and the sum of roots is negative. This gives that the remaining two eigenvalues of $ H $ are also negative.
\end{proof}

\section{Proof of Proposition~\ref{thmFIpr} (Full Implementation - Private goods)}
\label{proofFIpr}
	
\begin{proof} 
For the private goods problem in~\eqref{eqcppr}, the optimality conditions in~\eqref{eqKKTpr} are sufficient. Thus in order to prove that the corresponding allocation at Nash equilibrium is efficient, we show that at any Nash equilibrium $ {\overline{m}} = (\overline{y},\overline{q}) \in \mathcal{M} $, the allocation $ \big(\widehat{x}_i(\overline{m}) \big)_{i\in \mN} $ and prices $ \big(\widehat{p}_i(\overline{m})\big)_{i\in \mN} $ satisfy the optimality conditions as $ x^* $ and $\lambda_1^* $, respectively. Then using an invertibility argument we show existence and uniqueness of Nash equilibrium.
		
Using Proposition~\ref{thmconcpr}, at any Nash equilibrium $ \overline{m} $ we have: $ { \nabla_{m_i} } u_i(\overline{m}) = 0 $, $ \forall~i \in \mN $. This gives
\begin{subequations}
\begin{alignat}{2}
	\frac{\partial v_i(\widehat{x}_i(\overline{m}))}{\partial y_i} - \frac{\partial \widehat{t}_i(\overline{m})}{\partial y_i} &= 0, \quad &&\forall~i \in \mN, 
	\\
	\frac{\partial v_i(\widehat{x}_i(\overline{m}))}{\partial q_i^r} - \frac{\partial \widehat{t}_i(\overline{m})}{\partial q_i^r} &= 0, \quad &&\forall~r \in \mN,~i\in \mN.	
\end{alignat}
\end{subequations}
Using the definitions in~\eqref{eqallopr} and \eqref{eqtaxpr}, this becomes
\begin{subequations} \label{eqNE1pr}
\begin{gather}\label{eqNE1pr_a}
	v_i^\prime(\widehat{x}_i(\overline{m})) - \widehat{p}_i(\overline{m}_{-i}) + 2\xi(\overline{q}_i^i - \xi \overline{y}_i) = 0, \quad \forall~i \in \mN, 
	\\
	\label{eqNE1pr_b}
	\overline{q}_i^r =
	\left\{
	\begin{array}{ll}
	\xi \overline{y}_i  & \mbox{for } r = i, \\
	\xi \overline{y}_r  & \mbox{for } r \in \mN(i), \\
	\xi \overline{q}_{n(i,r)}^r & \mbox{for } r \notin  \mN(i) \text{ and } r \ne i,
	\end{array}
	\right. \quad \forall~i \in \mN. 
\end{gather}
\end{subequations}
		
For any distinct pair of vertexes $ i,r $, denote by $ \{ i, i_1,i_2,\ldots,i_{d(i,r)} = r \}$ the ordered vertexes in the shortest path between $ i $ and $ r $, where $ i_1 = n(i,r) \in \mN(i) $. Since the shortest path between $ i $ and $ r $ contains the shortest path between $ i_k $ and $ r $, for any $ k < d(i,r) $, we have $ n(i_k,r) = i_{k+1} $. Using the third sub-equation in~\eqref{eqNE1pr_b} repeatedly, replacing $ i $ by $ i_k $ gives,
\begin{equation}
	\overline{q}_i^r = \xi \, \overline{q}_{i_1}^r = \xi^2 \, \overline{q}_{i_2}^r = \cdots =  \xi^{d(i,r)-1} \, \overline{q}_{i_{d(i,r)-1}}^r. 	
\end{equation}
Now using the second sub-equation of~\eqref{eqNE1pr_b}, replacing $ i $ by $ i_{d(i,r)-1} $ and noting $ r \in \mN(i_{d(i,r)-1}) $, gives $ \overline{q}_{i_{d(i,r)-1}}^r = \xi \, \overline{y}_{r} $. This combined with the above equation gives that~\eqref{eqNE1pr_b} implies 
\begin{equation}\label{eqlempr}
	\overline{q}_i^r =
	\left\{
	\begin{array}{ll}
	\xi \overline{y}_i  & \mbox{\textup{for} } r = i, \\
	\xi^{d(i,r)} \overline{y}_r  & \mbox{\textup{for} } r \ne i,	
	\end{array}
	\right. \quad \forall~i \in \mN. 
\end{equation}
		
Using the above and then combining~\eqref{eqNE1pr_a} with~\eqref{eqallopr} and~\eqref{eqtaxpr_b} gives, $ \forall~i\in \mN $, 
\begin{subequations}
\begin{align} \label{eqNE2pr_a}
	v_i^\prime(\widehat{x}_i(\overline{m})) &= \widehat{p}_i(\overline{m}_{-i}), 
	\\ 
	\label{eqNE2pr_b}
	\widehat{x}_i(\overline{m}) &= \overline{y}_i - \frac{1}{N-1}\sum_{j \ne i} \overline{y}_j, 
	\\ 
	\label{eqNE2pr_c}
	\widehat{p}_i(\overline{m}_{-i}) &= \frac{1}{\delta} \sum_{j \in \mN} \overline{y}_j. 
\end{align}
\end{subequations}
\eqref{eqNE2pr_b} implies $ \sum_{i \in \mN} \widehat{x}_i(\overline{m}) = 0 $ and combining~\eqref{eqNE2pr_a} and~\eqref{eqNE2pr_c} gives $ v_i^\prime(\widehat{x}_i(\overline{m})) = \frac{1}{\delta} \sum_{j \in \mN} \overline{y}_j $, $ \forall~i\in \mN $. Thus, the allocation-price pair
\begin{equation}
	\left(\big(\overline{y}_i - \frac{1}{N-1}\sum_{j \ne i} \overline{y}_j\big)_{i \in \mN} \, , \, \frac{1}{\delta} \sum_{j \in \mN} \overline{y}_j\right)
\end{equation}
satisfy the optimality conditions,~\eqref{eqKKTpr}, as $ (x^*,\lambda_1^*) $. Since the optimality conditions are sufficient, the allocation at any Nash equilibrium $ \overline{m} $ is the efficient allocation $ x^* $.
		
For existence and uniqueness, consider the following set of linear equations that must be satisfied at any Nash equilibrium $ \overline{m} $,
\begin{subequations}
\begin{align}
	x_i^* &= \overline{y}_i - \frac{1}{N-1}\sum_{j \ne i} \overline{y}_j, \quad \forall~ i \in \mN, 
	\\
	\lambda_1^* &= \frac{1}{\delta} \sum_{j \in \mN} \overline{y}_j.
\end{align}
\end{subequations}
Here $ \left(\overline{y}_j \right)_{j\in\mN} $ are the variables and $ (x^*,\lambda_1^*) $ are fixed - since they are uniquely defined by the optimization~\eqref{eqcppr}. The above equations can be inverted to give the unique solution as,
\begin{equation}
	\overline{y}_i = \frac{N-1}{N} x_i^* + \frac{\delta \lambda_1^*}{N}, \quad \forall~i\in\mN. 
\end{equation}
Furthermore, using above and~\eqref{eqlempr}, the values for $ \big( \overline{q}_i^r \big)_{i,r\in\mN} $ can also be calculated uniquely. Since a solution for $ \overline{m} = (\overline{y},\overline{q}) $ in terms of $ x^*,\lambda_1^* $ exists, existence of Nash equilibrium is guaranteed. Also, since this solution is unique, there is a unique Nash equilibrium.
		
For Budget Balance, we have the following. From the above characterization, at Nash Equilibrium $ \widetilde{m} $ all tax terms from~\eqref{eqtaxpr_a}, other than $ \widehat{p}_i(\widetilde{m}_{-i}) \left(\widehat{x}_i(\widetilde{m}) - \dfrac{c_1}{N} \right) $, are zero. Furthermore, the prices are all equal to $ \lambda_1^* $ and allocations are equal to $ x_i^* $. Thus,
\begin{equation}
	\sum_{i \in \mN} \widehat{t}_i(\widetilde{m}) = \sum_{i \in \mN} \lambda_1^* \left(\widehat{x}_i(\widetilde{m}) - \dfrac{c_1}{N} \right) = \lambda_1^* \left( \sum_{i \in \mN} x_i^* - c_1 \right) = \lambda_1^* \cdot 0 = 0,
\end{equation}
since the efficient allocation $ x^* $ satisfies the constraint $ \sum_{i \in \mN} x_i^* = c_1 $.
\end{proof}

\section{Proof of Lemma~\ref{thmconcpub} (Concavity - Public Goods)}
\label{proofconcpub}
	
\begin{proof}
Since the allocation and tax functions are smooth and $ v_i(\cdot) $ is continuously double-differentiable, to establish concavity we show that the Hessian of $ u_i(m) $ w.r.t. $ m_i $ is negative definite i.e., $ H \prec 0 $. Once this is established, the optimization in~\eqref{eqBRdefpr} has a strictly concave objective and an unbounded constraint set. Thus it has a unique maximizer, defined by the first order derivative conditions.
		
The Hessian is of size $ (N+1) \times (N+1) $ and we have
\begin{subequations}	
\begin{align}
	H_{11} &= \frac{\partial^2 u_i(m)}{\partial y_i^2} = \dfrac{v_i^{\prime\prime}(\widehat{x}_i(m))}{N^2} - (2+\delta)\xi^2, 
	\\
	H_{(j+1)1} = H_{1(j+1)} &= \frac{\partial^2 u_i(m)}{\partial y_i \partial q_i^j} =
	\left\{
	\begin{array}{ll}
	0  & \mbox{for } j \in \mN,~j\ne i, \\
	2\xi & \mbox{for } j=i,
	\end{array}
	\right. 
	\\
	H_{(j+1)(j+1)} &= \frac{\partial^2 u_i(m)}{\partial (q_i^j)^2} = -2, \quad \forall~ j \in \mN,
	\\	
	H_{(r+1)(j+1)} &= \frac{\partial^2 u_i(m)}{\partial q_i^r \partial q_i^j} = 0 , \quad \forall~ j,r \in \mN,~j \ne r.
\end{align}
\end{subequations}
The characteristic equation, $ \textsf{Det}\left(H - xI\right) = 0 $, becomes
\begin{equation}
	\left(x+2\right)^{N-1} \Big( (x+2)(x-H_{11}) - 4\xi^2 \Big) = 0.
\end{equation}
This implies that $ N-1 $ eigenvalues of $ H $ are $ -2 $ and the remaining two eigenvalues satisfy $ x^2 + (2-H_{11})x + 2\delta\xi^2 - \frac{2}{N^2}v_i^{\prime\prime}(\widehat{x}_i(m)) = 0 $. Since $ H $ is a symmetric matrix, all its eigenvalues are real. Due to $ v_i^{\prime\prime}(\cdot) < 0 $, the product of roots in the above quadratic equation is positive and the sum of roots is negative. This gives that the remaining two eigenvalues of $ H $ are also negative.
\end{proof}

\section{Proof of Proposition~\ref{thmFIpub} (Full Implementation - Public goods)}
\label{proofFIpub}
	
\begin{proof}
For the public goods problem in~\eqref{eqcppub}, the optimality conditions in~\eqref{eqKKTpub} are sufficient. Thus in order to prove that the corresponding allocation at Nash equilibrium is efficient, we show that at any Nash equilibrium $ {\overline{m}} = (\overline{y},\overline{q}) \in \mathcal{M} $, the allocation $ \big(\widehat{x}_i(\overline{m}) \big)_{i\in \mN} $ and prices $ \big(\widehat{p}_i(\overline{m})\big)_{i\in \mN} $ satisfy the optimality conditions as $ x^* $ and $ {\mu^1}^* $, respectively. Then using an invertibility argument we show existence and uniqueness of Nash equilibrium.
		
Using Proposition~\ref{thmconcpub}, at any Nash equilibrium $ \overline{m} $ we have: $ { \nabla_{m_i} } u_i(\overline{m}) = 0 $, $ \forall~i \in \mN $. This gives
\begin{subequations}
\begin{alignat}{2}
	\frac{\partial v_i(\widehat{x}_i(\overline{m}))}{\partial y_i} - \frac{\partial \widehat{t}_i(\overline{m})}{\partial y_i} &= 0, \quad &&\forall~i \in \mN, 
	\\
	\frac{\partial v_i(\widehat{x}_i(\overline{m}))}{\partial q_i^r} - \frac{\partial \widehat{t}_i(\overline{m})}{\partial q_i^r} &= 0, \quad &&\forall~r \in \mN,~i\in \mN.	
\end{alignat}
\end{subequations}
Using the definitions in~\eqref{eqallopub} and \eqref{eqtaxpub}, this becomes, $ \forall~i \in \mN $,
\begin{subequations} \label{eqNE1pub}
\begin{gather}\label{eqNE1pub_a}
	\frac{1}{N} \big( v_i^\prime(\widehat{x}_i(\overline{m})) - \widehat{p}_i(\overline{m}_{-i}) \big) + 2\xi(\overline{q}_i^i - \xi \overline{y}_i) + \delta\xi(\overline{q}_{n(i,i)}^i - \xi \overline{y}_i) = 0, 
	\\
	\label{eqNE1pub_b}
	\overline{q}_i^r =
	\left\{
	\begin{array}{ll}
	\xi \overline{y}_i  & \mbox{for } r = i, \\
	\xi \overline{y}_r  & \mbox{for } r \in \mN(i), \\
	\xi \overline{q}_{n(i,r)}^r & \mbox{for } r \notin  \mN(i) \text{ and } r \ne i,
	\end{array}
	\right.. 
\end{gather}
\end{subequations}
		
For any distinct pair of vertexes $ i,r $, denote by $ \{ i, i_1,i_2,\ldots,i_{d(i,r)} = r \}$ the ordered vertexes in the shortest path between $ i $ and $ r $, where $ i_1 = n(i,r) \in \mN(i) $. Since the shortest path between $ i $ and $ r $ contains the shortest path between $ i_k $ and $ r $, for any $ k < d(i,r) $, we have $ n(i_k,r) = i_{k+1} $. Using the third sub-equation in~\eqref{eqNE1pub_b} repeatedly, replacing $ i $ by $ i_k $ gives,
\begin{equation}
	\overline{q}_i^r = \xi \, \overline{q}_{i_1}^r = \xi^2 \, \overline{q}_{i_2}^r = \cdots =  \xi^{d(i,r)-1} \, \overline{q}_{i_{d(i,r)-1}}^r. 	
\end{equation}
Now using the second sub-equation of~\eqref{eqNE1pub_b}, replacing $ i $ by $ i_{d(i,r)-1} $ and noting $ r \in \mN(i_{d(i,r)-1}) $, gives $ \overline{q}_{i_{d(i,r)-1}}^r = \xi \, \overline{y}_{r} $. This combined with the above equation gives that~\eqref{eqNE1pub_b} implies
\begin{equation}\label{eqlempub}
	\overline{q}_i^r =
	\left\{
	\begin{array}{ll}
	\xi \overline{y}_i  & \mbox{\textup{for} } r = i, \\
	\xi^{d(i,r)} \overline{y}_r  & \mbox{\textup{for} } r \ne i,	
	\end{array}
	\right. \quad \forall~i \in \mN. 
\end{equation}
		
Using the above and then combining~\eqref{eqNE1pub_a} with~\eqref{eqallopub} and~\eqref{eqtaxpub_b} gives, $ \forall~i\in \mN $, 
\begin{subequations}
\begin{align} \label{eqNE2pub_a}
	v_i^\prime(\widehat{x}_i(\overline{m})) &= \widehat{p}_i(\overline{m}_{-i}), 
	\\ 
	\label{eqNE2pub_b}
	\widehat{x}_i(\overline{m}) &= \frac{1}{N} \sum_{j \in \mN} \overline{y}_j , 
	\\ 
	\label{eqNE2pub_c}
	\widehat{p}_i(\overline{m}_{-i}) &= \delta(N-1) \left( \overline{y}_i - \frac{1}{N-1} \sum_{j \ne i} \overline{y}_j \right). 
\end{align}
\end{subequations}
\eqref{eqNE2pub_b} implies $ \widehat{x}_i(\overline{m}) = \widehat{x}_r(\overline{m}) $ for any $ i,r \in \mN $ and~\eqref{eqNE2pub_c} gives $ \sum_{r \in \mN} \widehat{p}_i(\overline{m}_{-i}) = 0 $. Thus, the allocation-price pair
\begin{equation}
	\left(\frac{1}{N} \sum_{j \in \mN} \overline{y}_j \, , \, \left( \delta(N-1) \left(\overline{y}_i - \frac{1}{N-1} \sum_{j \ne i} \overline{y}_j \right)\right)_{i \in \mN} \right)
\end{equation}
satisfy the optimality conditions,~\eqref{eqKKTpub}, as $ (x^*,{\mu^1}^*) $. Since the optimality conditions are sufficient, the allocation at any Nash equilibrium $ \overline{m} $ is the efficient allocation $ x^* $.
		
For existence and uniqueness, consider the following set of linear equations that must be satisfied at any Nash equilibrium $ \overline{m} $,
\begin{subequations}
\begin{align}
	x^* &= \frac{1}{N} \sum_{j \in \mN} \overline{y}_j, 
	\\
	{\mu^1_i}^* &= \delta(N-1) \left( \overline{y}_i - \frac{1}{N-1} \sum_{j \ne i} \overline{y}_j\right), \quad \forall~i \in \mN.
\end{align}
\end{subequations}
Here $ \left(\overline{y}_j \right)_{j\in\mN} $ are the variables and $ (x^*,{\mu^1}^*) $ are fixed - since they are uniquely defined by the optimization,~\eqref{eqcppub}. The above equations can be inverted to give the unique solution as,
\begin{equation}
	\overline{y}_i = x^* + \frac{{\mu_i^1}^*}{\delta N}, \quad \forall~i\in\mN. 
\end{equation}
Furthermore, using above and~\eqref{eqlempub}, the values for $ \big( \overline{q}_i^r \big)_{i,r\in\mN} $ can also be calculated uniquely. Since a solution for $ \overline{m} = (\overline{y},\overline{q}) $ in terms of $ x^*,{\mu^1}^* $ exists, existence of Nash equilibrium is guaranteed. Also, since this solution is unique, there is a unique Nash equilibrium.
		
For Budget Balance, we have the following. 	By the characterization from above we know that at Nash Equilibrium $ \widetilde{m} $, all tax terms from~\eqref{eqtaxpub_a}, other than $ \widehat{p}_i(\widetilde{m}_{-i}) \widehat{x}_i(\widetilde{m}) $, are zero. Furthermore, the prices are equal to $ {\mu_i^1}^* $ and each allocation is equal to $ x^* $. Thus,
\begin{equation}
	\sum_{i \in \mN} \widehat{t}_i(\widetilde{m}) = \sum_{i \in \mN} {\mu_i^1}^* x^* = x^* \sum_{i \in \mN} {\mu_i^1}^* = x^* \cdot 0 = 0,
\end{equation}
since the optimal dual variables $ \big({\mu_i^1}^*\big)_{i\in\mN} $ satisfy,~\eqref{eqKKTpub_b}, $ \sum_{i \in \mN} {\mu_i^1}^* = 0 $.
\end{proof}

\section{Proof of Proposition~\ref{thmcontracpr} (Contraction - Private goods)}
\label{proofcontracpr}
	
\begin{proof}
The game is contractive if the matrix norm of the Jacobian of best-response $ \beta = \big( \beta_i \big)_{i \in \mN} = \big( \breve{y}_i,\breve{q}_i \big)_{i \in \mN} $ is smaller than unity, i.e.,  $ \Vert \nabla \beta \Vert < 1 $. We use the row-sum norm for this, and in this proof verify specifically  the following set of conditions,
\begin{subequations}
\begin{alignat}{2}
	\label{eqBRypr}
	\sum_{r \in \mN,\,r  \ne i} \left( \left\vert \frac{\partial \breve{y}_i}{\partial y_r}  \right\vert + \sum_{j\in \mN} \left\vert \frac{\partial \breve{y}_i}{\partial q_r^j}  \right\vert \right) &< 1, \quad &&\forall~i \in \mN, 
	\\
	\label{eqBRqpr}
	\sum_{r \in \mN,\,r  \ne i} \left( \left\vert \frac{\partial \breve{q}_i^w}{\partial y_r}  \right\vert + \sum_{j\in \mN} \left\vert \frac{\partial \breve{q}_i^w}{\partial q_r^j}  \right\vert \right) &< 1, \quad &&\forall~w \in \mN,~\forall~i \in \mN. 		
\end{alignat}
\end{subequations}
The summation can be performed simply over the indexes $ r \in \mN(i) $ instead of $ r \ne i $ because our defined mechanism is distributed and hence the best-response of agent $ i $ depends only on $ \big( m_j \big)_{j \in \mN(i)} $. 
		
Consider any agent $ i \in \mN $, for the best-response $ \breve{q}_i $ we have 
\begin{equation}
	\breve{q}_i^w =
	\left\{
	\begin{array}{ll}
	\xi \breve{y}_i  & \mbox{for } w = i, \\
	\xi y_w  & \mbox{for } w \in \mN(i), \\
	\xi q_{n(i,w)}^w & \mbox{for } w \notin  \mN(i) \text{ and } w \ne i.
	\end{array}
	\right.
\end{equation}
Thus, by choosing $ \xi \in (0,1) $, all conditions within~\eqref{eqBRqpr} are satisfied where $ w \ne i $. Next, we verify conditions in~\eqref{eqBRypr}. Once this is done, then in conjunction with $ \xi \in (0,1) $, the conditions from~\eqref{eqBRqpr} with $ w=i $ are also automatically verified. 
		
For the best-response $ \breve{y}_i $, we have
\begin{align} \label{eqBRy2pr}
\breve{y}_i =  \frac{1}{N-1} \sum_{\substack{r \notin \mN(i) \\ r \ne i} } \frac{q_{n(i,r)}^r}{\xi^{d(i,r)-1}} + \frac{1}{N-1} \sum_{r \in \mN(i)} \frac{q_{r}^r}{\xi} 
+ \left( v_i^\prime \right)^{-1} \big( \widehat{p}_i(m_{-i}) \big).
\end{align} 
where $ \widehat{p}_i(m_{-i}) $ is defined in~\eqref{eqtaxpr_b}. Thus, we have,
\begin{align}\label{eqBRyderpr}
	\frac{\partial \breve{y}_i}{\partial q_{n(i,r)}^r} = 
	\left\{
	\begin{array}{ll}
	\dfrac{1}{\delta \xi} \dfrac{1}{v_i^{\prime\prime}(\cdot)} & \mbox{for } r=i,\\[2ex]
	\dfrac{1}{N-1} \dfrac{1}{\xi} + \dfrac{1}{\delta \xi} \dfrac{1}{v_i^{\prime\prime}(\cdot)}  & \mbox{for } r \in \mN(i),\\[2ex]
	\dfrac{1}{N-1} \dfrac{1}{\xi^{d(i,r)-1}} + \dfrac{1}{\delta \xi^{d(i,r)-1}} \dfrac{1}{v_i^{\prime\prime}(\cdot)} & \mbox{for } r \notin  \mN(i),~ r \ne i.
	\end{array}
	\right. 
\end{align}
where in each expression above $ v_i^{\prime\prime}(\cdot) $ is evaluated at $ \widehat{p}_i(m_{-i}) $. Also, in the notation used above, for any $ r \in \mN(i) $, $ n(i,r) = r $. All other partial derivative of $ \breve{y}_i $ are zero. With all this condition in~\eqref{eqBRypr} becomes, 
\begin{equation}\label{eqBRy3pr}
	\left\vert \dfrac{1}{\delta \xi} \dfrac{1}{v_i^{\prime\prime}(\cdot)}  \right\vert
	+ \left\vert 1 + \frac{N-1}{\delta v_i^{\prime\prime}(\cdot)} \right\vert \left( \sum_{r \in \mN(i)} \frac{1}{(N-1)\xi} \right) 
	+ \left\vert 1 + \frac{N-1}{\delta v_i^{\prime\prime}(\cdot)} \right\vert \Bigg( \sum_{\substack{r \notin \mN(i) \\ r \ne i} } \frac{1}{(N-1)\xi^{d(i,r)-1}} \Bigg)   < 1.
\end{equation}
To simplify the above, we utilize the upper bound from~\eqref{eqetapr}, $ v_i^{\prime\prime}(\cdot) \in (-\eta,-\frac{1}{\eta}) $. Set 
\begin{equation} \label{eqetaC1pr}
	\eta < \frac{\delta}{N-1},
\end{equation}
so that the expressions inside absolute value operator for the second and third terms on the LHS in~\eqref{eqBRy3pr} are guaranteed to be positive. With this,~\eqref{eqBRy3pr} becomes
\begin{equation} \label{eqBRy4pr}
	\left[ \frac{-1}{v_i^{\prime\prime}(\cdot)} \right] \left( \frac{1}{\xi} - \sum_{r \in \mN(i)} \frac{1}{\xi} - \sum_{\substack{r \notin \mN(i) \\ r \ne i}} \frac{1}{\xi^{d(i,r)-1}} \right) < \frac{\delta}{N-1} \left( N-1 - \left[  \sum_{r \in \mN(i)} \frac{1}{\xi} + \sum_{\substack{r \notin \mN(i) \\ r \ne i}} \frac{1}{\xi^{d(i,r)-1}}\right]  \right)
\end{equation}
Since any agent has at least one neighbor i.e., $ \vert \mN(i) \vert \ge 1 $, the LHS above is negative. For the RHS, note that the expression inside the square brackets has exactly $ N-1 $ terms and each term is of the form $ \xi^{-k} $, for some $ k \ge 1 $. Since $ \xi < 1 $, this gives that even the RHS is negative. Utilizing the lower bound from~\eqref{eqetapr}, a sufficient condition to verify~\eqref{eqBRy4pr} is $ \eta < \frac{N-1}{\delta} \left\vert \frac{C_i}{D_i} \right\vert  $, where $ C_i,D_i $ are the expression inside the curved bracket on the LHS and RHS of~\eqref{eqBRy4pr}, respectively.
Combining this with the condition in~\eqref{eqetaC1pr}, a sufficient condition for verifying~\eqref{eqBRypr} is 
\begin{equation}
	\eta < \min\left(\frac{\delta}{N-1} \, , \, \frac{N-1}{\delta}\left\vert \frac{C_i}{D_i} \right\vert\right), \quad \forall~i \in \mN.
\end{equation}
Without any further tuning of parameters $ \xi,\delta $, the proof is complete as long as $ \eta $ satisfies above. However, in our model we would like to accommodate any value of $ \eta > 1 $ and this requires tuning of parameters $ \xi,\delta $. Set 
\begin{equation}
	\delta = (N-1)\sqrt{\underset{i \in \mN}{\min} \left\vert \dfrac{C_i}{D_i} \right\vert} > 0,
\end{equation}
in this to get the sufficient condition as $ \eta^2 < \underset{i \in \mN}{\min} \left\vert \dfrac{C_i}{D_i} \right\vert $, i.e.,
\begin{equation}
	\eta^2 < \underset{i \in \mN}{\min} \left\{ \left. \left \vert {\dfrac{1}{\xi} - \displaystyle\sum_{r \in \mN(i)} \dfrac{1}{\xi} - \displaystyle\sum_{\substack{r \notin \mN(i) \\ r \ne i}} \dfrac{1}{\xi^{d(i,r)-1}} } \right \vert 
	\middle/  
	\left \vert {N-1 - \left[  \displaystyle\sum_{r \in \mN(i)} \dfrac{1}{\xi} + \displaystyle\sum_{\substack{r \notin \mN(i) \\ r \ne i}} \dfrac{1}{\xi^{d(i,r)-1}}\right]} \right \vert \right. \right\}.
\end{equation}
We want to select $ \xi \in (0,1) $ such that the RHS above can be made arbitrarily large. For this, first note that,  for any $ i \in \mN $ the numerator of the RHS is bounded away from zero for $ \xi $ in the neighborhood of $ 1 $. Second, the denominator can be made arbitrarily close to $ 0 $ by choosing $ \xi $ close enough to $ 1 $. This can be seen by rewriting 
\begin{equation}
	D_i = N-1 - \left[  \sum_{r \in \mN(i)} \frac{1}{\xi} + \sum_{\substack{r \notin \mN(i) \\ r \ne i}} \frac{1}{\xi^{d(i,r)-1}}\right] = \sum_{r \in \mN(i)} \left( 1 - \frac{1}{\xi} \right)  + \sum_{\substack{r \notin \mN(i) \\ r \ne i}} \left( 1 - \frac{1}{\xi^{d(i,r)-1}} \right), 
\end{equation}
where for any given $ k \ge 1,\epsilon > 0 $, choose $ \xi \in \left( \left(\frac{1}{1+\epsilon}\right)^{\nicefrac{1}{k}},1 \right) $ to have $ \vert 1 - \xi^{-k} \vert < \epsilon $. Finally, it is clear from above that the denominator $ D_i $ can be made arbitrarily small concurrently for all $ i \in \mN $.
		
That the game is not supermodular follows from the first sub-equation of~\eqref{eqBRyderpr}, which implies that the best-response $ \breve{y}_i $ is decreasing w.r.t. $ q_{n(i,i)}^i $. Also, convergence of every learning dynamic within the ABR class is guaranteed by Fact~\ref{factabr}.
\end{proof}

\section{Proof of Proposition~\ref{thmcontracpub} (Contraction - Public goods)}
\label{proofcontracpub}
	
\begin{proof}
The game is contractive if the matrix norm of the Jacobian of best-response $ \beta = \big( \beta_i \big)_{i \in \mN} = \big( \breve{y}_i,\breve{q}_i \big)_{i \in \mN} $ is smaller than unity, i.e.,  $ \Vert \nabla \beta \Vert < 1 $. We use the row-sum norm for this, and in this proof verify specifically  the following set of conditions,
\begin{subequations}
\begin{alignat}{2}
	\label{eqBRypub}
	\sum_{r \in \mN,\, r  \ne i} \left( \left\vert \frac{\partial \breve{y}_i}{\partial y_r}  \right\vert + \sum_{j\in \mN} \left\vert \frac{\partial \breve{y}_i}{\partial q_r^j}  \right\vert \right) &< 1, \quad &&\forall~i \in \mN, 
	\\
	\label{eqBRqpub}
	\sum_{r \in \mN,\, r  \ne i} \left( \left\vert \frac{\partial \breve{q}_i^w}{\partial y_r}  \right\vert + \sum_{j\in \mN} \left\vert \frac{\partial \breve{q}_i^w}{\partial q_r^j}  \right\vert \right) &< 1, \quad &&\forall~w \in \mN,~\forall~i \in \mN. 		
\end{alignat}
\end{subequations}
The summation can be performed simply over the indexes $ r \in \mN(i) $ instead of $ r \ne i $ because our defined mechanism is distributed and hence the best-response of agent $ i $ depends only on $ \big( m_j \big)_{j \in \mN(i)} $. 
		
Consider any agent $ i \in \mN $, for the best-response $ \breve{q}_i $ we have 
\begin{equation}
	\breve{q}_i^w =
	\left\{
	\begin{array}{ll}
	\xi \breve{y}_i  & \mbox{for } w = i, \\
	\xi y_w  & \mbox{for } w \in \mN(i), \\
	\xi q_{n(i,w)}^w & \mbox{for } w \notin  \mN(i) \text{ and } w \ne i.
	\end{array}
	\right.
\end{equation}
Thus, by choosing $ \xi \in (0,1) $, all conditions within~\eqref{eqBRqpub} are satisfied where $ w \ne i $. Next, we verify conditions in~\eqref{eqBRypub}. Once this is done, then in conjunction with $ \xi \in (0,1) $, the conditions from~\eqref{eqBRqpub} with $ w=i $ are also automatically verified. 
		
For the best-response $ \breve{y}_i $, we have the following relation
\begin{subequations}
\begin{align}
	\frac{1}{N} \big( v_i^\prime(\widehat{x}_i(m)) - \widehat{p}_i(m_{-i}) \big) + 2\xi(\breve{q}_i^i - \xi \breve{y}_i) + \delta\xi(q_{n(i,i)}^i - \xi \breve{y}_i) &= 0, 
	\\
	\Rightarrow~~ 	\frac{1}{N} \big( v_i^\prime(\widehat{x}_i(m)) - \widehat{p}_i(m_{-i}) \big) + \delta\xi(q_{n(i,i)}^i - \xi \breve{y}_i) &= 0. 
\end{align}
\end{subequations}
In the above relation, $ \widehat{x}_i(m) $ is evaluated at $ \breve{y}_i $ instead of $ y_i $. Also, this relation implicitly defines $ \breve{y}_i $. Differentiating this equation w.r.t. $ \big( q_{n(i,r)}^r \big)_{r \in \mN} $ gives
\begin{subequations}
\begin{alignat}{2} 
	\dfrac{v_i^{\prime\prime}(\widehat{x}_i(m))}{N^2} \dfrac{\partial \breve{y}_i}{\partial q_{n(i,i)}^r} - \dfrac{\delta(N-1)}{N\xi} + \delta\xi\left( 1 - \xi \dfrac{\partial \breve{y}_i}{\partial q_{n(i,i)}^r} \right) &= 0, \quad  ~ &&r=i, 
	\\
	\dfrac{v_i^{\prime\prime}(\widehat{x}_i(m))}{N^2} \left(\dfrac{\partial \breve{y}_i}{\partial q_{n(i,r)}^r} + \frac{1}{\xi} \right) + \dfrac{\delta}{N\xi} - \delta\xi^2\dfrac{\partial \breve{y}_i}{\partial q_{n(i,r)}^r} &= 0, \quad   \forall~&&r \in \mN(i),
	\\
	\dfrac{v_i^{\prime\prime}(\widehat{x}_i(m))}{N^2} \left(\dfrac{\partial \breve{y}_i}{\partial q_{n(i,r)}^r} + \frac{1}{\xi^{d(i,r)-1}} \right) + \dfrac{\delta}{N\xi^{d(i,r)-1}} - \delta\xi^2\dfrac{\partial \breve{y}_i}{\partial q_{n(i,r)}^r} &= 0, \quad  \forall~&&r \notin  \mN(i),~ r \ne i,
\end{alignat} 
\end{subequations}
which implies
\begin{align} \label{eqBRy2pub}
	\frac{\partial \breve{y}_i}{\partial q_{n(i,r)}^r} = \frac{1}{\dfrac{v_i^{\prime\prime}(\widehat{x}_i(m))}{N^2} - \delta\xi^2} \times 
	\left\{
	\begin{array}{ll}
	\dfrac{\delta(N-1)}{N\xi} - \delta\xi & \mbox{for } r=i, \\[2ex]
	-\dfrac{\delta}{N\xi} - \dfrac{v_i^{\prime\prime}(\widehat{x}_i(m))}{N^2 \xi} & \mbox{for } r \in \mN(i),\\[2ex]
	-\dfrac{\delta}{N\xi^{d(i,r)-1}} - \dfrac{v_i^{\prime\prime}(\widehat{x}_i(m))}{N^2 \xi^{d(i,r)-1}} & \mbox{for } r \notin  \mN(i),~ r \ne i.
	\end{array}
	\right. 
\end{align} 
In the notation used above, for any $ r \in \mN(i) $, $ n(i,r) = r $. All other partial derivative of $ \breve{y}_i $ are zero. With all this condition in~\eqref{eqBRypub} becomes, 
\begin{gather}
	\nonumber 
	\left\vert \frac{\delta(N-1)}{N\xi} - \delta\xi  \right\vert
	+ \left\vert -\dfrac{\delta}{N} - \dfrac{v_i^{\prime\prime}(\widehat{x}_i(m))}{N^2} \right\vert \left( \sum_{r \in \mN(i)} \frac{1}{\xi} \right) 
	+ \left\vert -\dfrac{\delta}{N} - \dfrac{v_i^{\prime\prime}(\widehat{x}_i(m))}{N^2} \right\vert \Bigg(  \sum_{\substack{r \notin \mN(i) \\ r \ne i} } \frac{1}{\xi^{d(i,r)-1}} \Bigg)
	\\ \label{eqBRy3pub}
	< \delta\xi^2 - \frac{v_i^{\prime\prime}(\widehat{x}_i(m))}{N^2}.
\end{gather}
We impose the condition
\begin{equation}\label{eqxipub}
	\xi \in \left( \sqrt{\frac{N-1}{N}} , 1 \right)
\end{equation}
so that the expression inside the first absolute value term in above is negative. To simplify the other expressions containing absolute value, we utilize the lower bound from~\eqref{eqetapr}, $ v_i^{\prime\prime}(\cdot) \in (-\eta,-\frac{1}{\eta}) $. Set 
\begin{equation} \label{eqetaC1pub}
	\eta < N\delta,
\end{equation}
so that the remaining expressions inside absolute value in~\eqref{eqBRy3pub} are guaranteed to be negative. With this,~\eqref{eqBRy3pub} becomes
\begin{multline} \label{eqBRy4pub} 
	\Big[ - v_i^{\prime\prime}(\widehat{x}_i(m))\Big] \left( 1 + \sum_{r \in \mN(i)} \frac{1}{\xi} + \sum_{\substack{r \notin \mN(i) \\ r \ne i}} \frac{1}{\xi^{d(i,r)-1}} \right) 
	\\ 
	>
	N\delta \left( \frac{1}{\xi} \left[ \sum_{\substack{r \notin \mN(i) \\ r \ne i}} \frac{1}{\xi^{d(i,r)-2}} - (N - \vert \mN(i) \vert  - 1)\right]  + N \xi (1-\xi)  \right),
\end{multline}
where $ N - \vert \mN(i) \vert  - 1  $ is the number of agents in the system except agent $ i $ and all his/her neighbors in $ \mN(i) $. Clearly the LHS above is positive. For any $ r \in \mN(i) $ and $ r \ne i $, we have $ d(i,r) \ge 2 $. On the RHS, inside the square brackets there are exactly $ N - \vert \mN(i) \vert -1 $ terms in the summation and each term is of the form $ \xi^{-k} $, for some $ k \ge 0 $. Since $ \xi < 1 $, this gives that even the RHS is positive. Utilizing the upper bound from~\eqref{eqetapr}, a sufficient condition to verify~\eqref{eqBRy4pub} is $ \eta < \dfrac{1}{N\delta} \dfrac{C_i}{D_i}  $, where $ C_i,D_i $ are the expression inside the curved bracket on the LHS and RHS of~\eqref{eqBRy4pub}, respectively. Combining this with the condition in~\eqref{eqetaC1pub}, a sufficient condition for verifying~\eqref{eqBRypub} is 
\begin{equation}
	\eta < \min\left(N\delta \, , \, \frac{1}{N\delta} \frac{C_i}{D_i} \right), \quad \forall~i\in \mN.
\end{equation}
Without any further tuning of parameters $ \xi,\delta $, the proof is complete as long as $ \eta $ satisfies above. However, in our model we would like to accommodate any value of $ \eta > 1 $ and this requires tuning of parameters $ \xi,\delta $. Set 
\begin{equation}
	\delta = \frac{1}{N}\sqrt{\underset{i\in\mN}{\min}\left(\frac{C_i}{D_i}\right) } > 0,
\end{equation} 
in this to get the sufficient condition as $ \eta^2 < \underset{i\in\mN}{\min} \left(\dfrac{C_i}{D_i}\right) $, i.e.,
\begin{multline}
	\eta^2 < \underset{i\in\mN}{\min} \Bigg\{ 
	\left( 1 + \sum_{r \in \mN(i)} \frac{1}{\xi} + \sum_{\substack{r \notin \mN(i) \\ r \ne i}} \frac{1}{\xi^{d(i,r)-1}} \right) 
	\\
	\left.
	\middle/  
	\left( \frac{1}{\xi} \left[ \sum_{\substack{r \notin \mN(i) \\ r \ne i}} \frac{1}{\xi^{d(i,r)-2}} - (N - \vert \mN(i) \vert  - 1)\right]  + N \xi (1-\xi) \right) 
	\right. \Bigg \}.
\end{multline}
We want to select $ \xi $  such that the RHS above can be made arbitrarily large, whilst satisfying~\eqref{eqxipub}. For this, firstly note that, for any $ i \in \mN $ the numerator of the RHS is greater than $ 1 $, hence it is bounded away from zero. Secondly, the denominator can be made arbitrarily close to $ 0 $ by choosing $ \xi $ close enough to $ 1 $. This can be seen by rewriting
\begin{subequations}
\begin{gather}
	\sum_{\substack{r \notin \mN(i) \\ r \ne i}} \frac{1}{\xi^{d(i,r)-2}} - (N - \vert \mN(i) \vert  - 1) 
	= 
	\sum_{\substack{r \notin \mN(i) \\ r \ne i}} \left( \frac{1}{\xi^{d(i,r)-2}} - 1 \right) 
	=
	\sum_{\substack{r \in \mN \\ d(i,r) \ge 3}} \left( \frac{1}{\xi^{d(i,r)-2}} - 1 \right),
	\\
	\Rightarrow~~D_i = \frac{1}{\xi} \sum_{\substack{r \in \mN \\ d(i,r) \ge 3}} \left( \frac{1}{\xi^{d(i,r)-2}} - 1 \right)  + N\xi(1-\xi),
\end{gather}
\end{subequations}
where for any given $ k \ge 1,\epsilon > 0 $, choose $ \xi \in \left( \left(\frac{1}{1+\epsilon}\right)^{\nicefrac{1}{k}},1 \right) $ to have $ \left( \xi^{-k} - 1 \right) < \epsilon $. Note that this is consistent with~\eqref{eqxipub}. The remaining term $ N\xi(1-\xi) $ can also made made arbitrarily small by choosing $ \xi $ close enough to $ 1 $. Finally, it is clear from above that the denominator $ D_i $ can made arbitrarily small concurrently for all $ i \in \mN $.
		
The fact that the game is supermodular follows from the preceding analysis, where the parameters are chosen such that each expression in~\eqref{eqBRy2pub} is positive. Which implies that the best-response is increasing in every message of other agents.
\end{proof}

\end{document}